\documentclass[11pt]{article}
\usepackage{fullpage}

\usepackage{times}
\usepackage{comment,amsfonts,amssymb,amsmath,amsthm,graphicx,algorithm,algorithmic}
\newcommand{\commentout}[1]{}

\ifx\pdftexversion\undefined
\usepackage[colorlinks,linkcolor=black,filecolor=black,citecolor=black,urlco
lor=black,pdfstartview=FitH]{hyperref}
\else
\usepackage[colorlinks,linkcolor=blue,filecolor=blue,citecolor=blue,urlcolor
=blue,pdfstartview=FitH]{hyperref}
\fi

\newcommand{\alert}[1]{\textbf{\color{red}
[[[#1]]]}\marginpar{\textbf{\color{red}**}}\typeout{ALERT:
\the\inputlineno: #1}}

\def\hE{\hat{E}}
\def\hX{\hat{X}}
\def\cP{{\cal P}}

\def\MathF{\hbox{\rm I\kern-2pt F}}
\def\MathP{\hbox{\rm I\kern-2pt P}}
\def\MathR{\hbox{\rm I\kern-2pt R}}
\def\MathZ{\hbox{\sf Z\kern-4pt Z}}
\def\MathN{\hbox{\rm I\kern-2pt I\kern-3.1pt N}}
\def\MathC{\hbox{\rm \kern0.7pt\raise0.8pt\hbox{\footnotesize I}
\kern-4.2pt C}}
\def\MathQ{\hbox{\rm I\kern-6pt Q}}

\newcommand{\polylog}{{\rm polylog}}

\newcommand{\R}{\mathbb{R}}

\newcommand{\etal}{\emph{et. al. }}

\newcommand{\E}{{\mathbb{E}}}

\newcommand{\mommit}[1]{}
\newcommand{\namedref}[2]{\hyperref[#2]{#1~\ref*{#2}}}

\newtheorem{theorem}{Theorem}
\newtheorem{lemma}{Lemma}
\newtheorem{corollary}[lemma]{Corollary}
\newtheorem{remark}{Remark}
\newtheorem{claim}[lemma]{Claim}

\usepackage{pdfsync}
\usepackage{authblk}
\newcommand{\ex}{{\cal EXP}}

\def\chS{\hat{\cal S}}
\def\chP{\hat{\cal P}}
\def\cP{{\cal P} }
\def\chU{\hat{\cal U}}
\def\exp{\mathit{exp}}

\def\tO{\tilde{O}}
\def\nin{\not \in}

\def\etal{{et al.~}}

\def\Rad{\mathit{Rad}}

\def\exp{\mathit{exp}}

\def\deg{\mathit{deg}}

\def\cS{{\cal S}}
\def\hC{{\hat C}}
\newcommand{\half}{\frac{1}{2}}
\def\eps{\epsilon}

\begin{document}
\title{Efficient Algorithms for Constructing Very Sparse Spanners and Emulators\footnote{A preliminary version \cite{EN17} of this paper appeared in SODA'17.}}
\author[1]{Michael Elkin\thanks{This research was supported by the ISF grant No. (724/15).}}
\author[1]{Ofer Neiman\thanks{Supported in part by ISF grant No. (523/12) and by BSF grant No. 2015813.}}

\affil[1]{Department of Computer Science, Ben-Gurion University of the Negev,
Beer-Sheva, Israel. Email: \texttt{\{elkinm,neimano\}@cs.bgu.ac.il}}

\date{}
\maketitle
\begin{abstract}
Miller et al. \cite{MPVX15} devised a distributed\footnote{They actually showed a PRAM algorithm. The distributed algorithm with these properties is implicit in \cite{MPVX15}.} algorithm in the CONGEST model, that given a parameter $k = 1,2,\ldots$, constructs an $O(k)$-spanner of an input unweighted $n$-vertex graph with $O(n^{1+1/k})$ expected edges in $O(k)$ rounds of communication. In this paper we improve the result of \cite{MPVX15}, by showing a $k$-round distributed algorithm in the same model, that constructs a $(2k-1)$-spanner with $O(n^{1+1/k}/\eps)$ edges,  with probability $1- \eps$, for any  $\eps>0$. Moreover, when $k = \omega(\log n)$, our algorithm produces (still in $k$ rounds) {\em ultra-sparse} spanners, i.e., spanners of size $n(1+ o(1))$, with probability $1- o(1)$. To our knowledge, this is the first distributed algorithm in the CONGEST or in the PRAM models  that constructs spanners or skeletons (i.e.,  connected spanning subgraphs) that sparse.
Our algorithm can also be implemented in linear time in the standard centralized model, and for large $k$, it provides spanners that are sparser than any other spanner given by a known (near-)linear time algorithm.


We also devise improved bounds (and algorithms realizing these bounds) for $(1+\epsilon,\beta)$-spanners and emulators. In particular, we show that for any unweighted $n$-vertex graph and any $\epsilon > 0$, there exists a $(1+ \epsilon, ({{\log\log n} \over \epsilon})^{\log\log n})$-emulator with $O(n)$ edges. All previous constructions of $(1+\epsilon,\beta)$-spanners and emulators employ a superlinear number of edges, for all choices of parameters.

Finally, we provide some applications of our results to approximate shortest paths' computation in unweighted graphs.
\end{abstract}


\section{Introduction}

\subsection{Setting, Definitions}

We consider unweighted undirected $n$-vertex graphs $G = (V,E)$. For a parameter $\alpha \ge 1$, a subgraph $H = (V,E')$, $E' \subseteq E$, is called an {\em $\alpha$-spanner} of $G$, if for every pair $u,v \in V$ of vertices, we have $d_H(u,v) \le \alpha \cdot d_G(u,v)$.
Here $d_G(u,v)$ (respectively, $d_H(u,v)$) stands for the distance between $u$ and $v$ in $G$ (resp., in $H$).  The parameter $\alpha$ is called the {\em stretch} of the spanner $H$. More generally, if for a pair of parameters $\alpha \ge 1$, $\beta \ge 0$, for every pair $u,v \in V$ of vertices, it holds that $d_H(u,v) \le \alpha \cdot d_G(u,v) + \beta$, then the subgraph $H$ is said to be an {\em $(\alpha,\beta)$-spanner} of $G$. Particularly important is the case $\alpha =1+\eps$, for some small $\eps > 0$. Such spanners are called {\em near-additive}.  If $H= (V,E'',\omega)$, where $\omega: E'' \rightarrow \R^+$, is not a subgraph of $G$, but nevertheless satisfies that for every pair $u,v \in V$ of original vertices, $d_G(u,v) \le d_H(u,v) \le (1+\eps) d_G(u,v) + \beta$, then $H$ is called a {\em near-additive $\beta$-emulator} of $G$, or a {\em $(1+\eps,\beta)$-emulator} of $G$.

Graph spanners have been introduced in \cite{Awe85,PS89,PU89a}, and have been intensively studied ever since \cite{ADDJS93,ABCP93,Coh93,ACIM99,DHZ00,BS03,E04,E06,EZ06,TZ06,W06,E07,Pet09,DGPV08,P10,BW15,MPVX15,AB16}. They were found useful for computing approximately shortest paths \cite{ABCP93,Coh93,E04,EZ06}, routing \cite{PU89b}, distance oracles and labeling schemes \cite{P99,TZ01,EP15}, synchronization \cite{Awe85}, and in other applications.

The simplest and most basic algorithm for computing a multiplicative $\alpha$-spanner, for a parameter $\alpha \ge 1$, is the {\em greedy} algorithm \cite{ADDJS93}.
The algorithm starts with an empty spanner, and examines the edges of the input graph $G = (V,E)$ one after another. It tests if there is a path in $H$ of length at most $\alpha$ between the endpoints $u$ and $v$ of $e$. If it is not the case, the edge is inserted into the spanner. Otherwise the edge is dropped.

It is obvious that the algorithm produces an $\alpha$-spanner $H$. Moreover, the spanner  $H$ has no cycles of length $\alpha + 1$ or less, i.e., the {\em girth} of $H$, denoted $g(H)$, satisfies $g(H) \ge \alpha + 2$. Denote $m = m(n,g)$ the maximum number of edges that a girth-$g$ $n$-vertex graph may contain. It follows that $|H| \le m(n,\alpha+2)$. The function $m(n,g)$ is known to be at most $n^{1+ {2 \over {g-2}}}$,
 when $g \le 2 \log_2 n$, and for larger $g$ (i.e., for $m \le 2n$), it is given by $m(n,g) \le n(1 + (1+o(1)) {{\ln (p+1)} \over g})$, where $p = m- n$, \cite{AHL02,BR10}. These bounds are called ``Moore's bounds for irregular graphs", or shortly, {\em (generalized) Moore's bounds}.

Any construction of multiplicative $\alpha$-spanners for $n$-vertex graphs with at most $m'(n,\alpha+2)$ edges implies an upper bound $m(n,\alpha+2) \le m'(n,\alpha+2)$ for the function $m(n,g)$. (As running the construction on the extremal girth-$(\alpha+2)$ $n$-vertex graph can eliminate no edge.) Hence
 the greedy algorithm produces multiplicative spanners with optimal tradeoff between stretch and number of edges. (See also \cite{FS16}.) However, the greedy algorithm  is problematic from algorithmic perspective. In the centralized model of computation, the best-known implementation of it \cite{RZ04} requires $O(\alpha  \cdot n^{2 + 1/\alpha})$ time.
 Moreover, the greedy algorithm is inherently sequential, and as such, it is generally hard\footnote{In the sequel we discuss a distributed setting, specifically, the LOCAL model, in which a relatively efficient implementation of the greedy is known.}  to implement it in distributed and parallel models of computation.

In the distributed model \cite{P00} we have processors residing in vertices of the graph. The processors communicate with their graph neighbors in synchronous rounds. In each round, messages of bounded length can be sent. (This is the assumption of the CONGEST model. In the LOCAL model, messages' size is arbitrary.)  The running time of an algorithm is this model is the number of rounds that it runs. By "parallel" model we mean here PRAM EREW model \cite{R93}, and we are interested in algorithms with small {\em running time} (aka {\em depth}) and {\em work} complexities. (The latter measures the overall number of operations performed by all processors.)

Dubhashi \etal \cite{DMPRS03} devised a distributed implementation of the greedy algorithm in the LOCAL model of distributed computation.
Their algorithm runs in $O(\alpha \cdot \log^2 n)$ rounds, i.e., suboptimal by a factor of $\log^2 n$. Moreover, it collects graph topology to depth $O(\alpha)$, and conducts heavy local computations. To our knowledge, there is no distributed-CONGEST or PRAM implementation of the greedy algorithm known. There is also no known efficient\footnote{By ``efficient" centralized algorithm in this paper we mean an algorithm with running time close to $O(|E|)$. By efficient distributed or parallel algorithm we mean an algorithm that runs in polylogarithmic, or nearly-polylogarithmic, time.}  centralized, distributed-CONGEST, or PRAM algorithm that constructs {\em ultra-sparse} spanners, i.e., spanners with $n + o(n)$ edges.

In the distributed and parallel settings it is often enough to compute a sparse {\em skeleton}  $H$ of the input graph $G$, where a {\em skeleton} is a connected subgraph that spans all the vertices of $G$, i.e., the stretch requirement is dropped. Dubhashi \etal \cite{DMPRS03} devised a distributed-LOCAL algorithm  that computes ultra-sparse skeletons of size $m(n,\alpha) \le n + O(n \cdot {{\log n} \over \alpha})$ in $O(\alpha)$ rounds. Like their algorithm for constructing spanners, this algorithm also collects topologies to depth $O(\alpha)$, and involves heavy local computations. To our knowledge, no efficient distributed-CONGEST or PRAM algorithm for computing {\em ultra}-sparse skeletons is currently known. In this paper we devise the first such algorithms.

\subsection{Prior Work and Our Results}

In the centralized model of computation the best-known efficient algorithm for constructing multiplicative spanners (for unweighted graphs) is due to Halperin and Zwick \cite{HZ96}. Their deterministic
algorithm, for an integer parameter $k \ge 1$, computes a $(2k-1)$-spanner with $n^{1+1/k} + n$ edges in $O(|E|)$ time.
(Their result improved previous pioneering work  by \cite{PS89,Coh93}.) Note that their bound on the number of edges is always at least $2n$,
 i.e., in the range $k = \Omega(\log n)$ it is very far from Moore's bound.

Our centralized randomized algorithm computes (with probability close to 1), a $(2k-1)$-spanner with $n\cdot (1 + O({{\log n} \over k}))$ edges in $O(|E|)$ time, whenever $k = \Omega(\log n)$. Note that when $k=\omega(\log n)$, the number of edges is $n(1+o(1))$, i.e., in this range the algorithm computes an ultra-sparse spanner in $O(|E|)$ time.
Moreover, whenever $k \le n^{1-\eps}$, for any constant $\eps>0$, up to a constant factor in the lower-order term, our bound matches Moore's bound. In fact, our algorithm and its analysis can be viewed as an alternative proof of (a slightly weaker version of) the generalized Moore's bound. Note that it is not the case for the greedy algorithm and its implementations \cite{ADDJS93,RZ04,DMPRS03}: the analysis of these algorithms relies on  Moore's bounds, but these algorithms cannot be used to derive them.

Another variant of our algorithm, which works for any $k \ge 2$, computes with high probability a $(2k-1)$-spanner with $n^{1+1/k}(1 + O({{\log k} \over k}))$ edges, in $\tO(k|E|)$ time.\footnote{As usual, $\tO(f(n))$ stands for $O(f(n) \polylog f(n))$.} In particular, for the range $k\ge {{2\ln n} \over {\ln\ln n}}$ the number of edges in our spanner is $n^{1+1/k}+o(n)$, improving the result of \cite{HZ96} (albeit with a somewhat worse running time for ${{2\ln n} \over {\ln\ln n}}\le k\le\log n$). Note that for any $k\ge 2$ we have $O(n^{1+1/k})$ edges.

Yet another variant of our algorithm computes a $(2k-1)$-spanner with $O(n^{1+1/k})$ edges, in expected $O(|E|)$ time.

In the distributed-CONGEST and PRAM models, efficient algorithms for computing linear-size spanners were given in \cite{Pet09,MPVX15}. Specifically, \cite{MPVX15} devised an $O(k)$-round distributed-CONGEST randomized algorithm for constructing $O(k)$-spanner (with high probability) with expected $O(n^{1+1/k})$ edges. In the PRAM model their algorithm has depth $O(k \log^* n)$ and work $O(|E|)$.
There are also $k$-round distributed-CONGEST randomized algorithms for constructing $(2k-1)$-spanner with expected $O(k \cdot n^{1+1/k})$ edges \cite{BS07,E06}.
It is known that at least $k$ rounds are required for this task, under Erd\H{o}s' girth conjecture \cite{E06,DGPV08}.

Our randomized algorithm  uses $k$ rounds in the distributed-CONGEST model, and with probability at least $1-\epsilon$ it constructs a $(2k - 1)$-spanner with $O(n^{1+1/k}/\epsilon)$ edges (for any desired, possibly sub-constant, $\epsilon>0$).
 In the PRAM model the depth and work complexities of our algorithm  are the same as in \cite{MPVX15}. Furthermore, when $k\ge \log n$ we can bound the number of edges by $n\cdot(1 + O({{\log n} \over \epsilon\cdot k}))$, again matching Moore's bound up to a constant factor in the lower-order term.

This result improves the previous state-of-the-art in the entire range of parameters. In particular, it is also the first efficient algorithm in the distributed-CONGEST or PRAM models that constructs an ultra-sparse skeleton. Specifically, in $\tO(\log n)$ time it computes an $\tO(\log n)$-spanner with $n(1 + o(1))$ edges, with probability $1 - o(1)$.

We also use our algorithm for unweighted graphs to devise an improved algorithm for {\em  weighted} graphs as well. Specifically, our algorithm computes $(2k-1)(1+\eps)$-spanner with $O(n^{1+1/k} \cdot (\log k)/\eps)$ edges, within expected $O(|E|)$ time. See Theorem \ref{thm:wt}, and the discussion that follows it, for further details.

\subsection{Near-Additive Spanners and Emulators}

It was shown in \cite{EP04} that for any $\eps > 0$ and $\kappa =1,2,\ldots$, and any (unweighted) $n$-vertex graph $G = (V,E)$, there exists a $(1+\eps,\beta)$-spanner with $O(\beta \cdot n^{1+1/\kappa})$ edges, where $\beta = \beta(\kappa,\eps) \le O({{\log \kappa} \over \eps})^{\log \kappa}$.
Additional algorithms for constructing such spanners were later given in \cite{E04,EZ06,TZ06,DGPV08,Pet09,P10}.
Abboud and Bodwin \cite{AB16} showed that multiplicative error of $1+\eps$ in \cite{EP04}'s theorem cannot be eliminated, while still keeping a constant (i.e., independent of $n$) additive error $\beta$, and more recently \cite{ABP17} showed that any such $(1+\eps,\beta)$-spanner of size $O(n^{1+1/\kappa-\delta})$, $\delta>0$, has $\beta=\Omega\left(\frac{1}{\epsilon\cdot\log \kappa}\right)^{\log \kappa-1}$.

 In the regime of constant $\kappa$, the bound of \cite{EP04} remains the state-of-the-art. Pettie \cite{Pet09} showed that one can construct a $(1+\eps,\beta)$-spanner with  $O(n \log\log (\eps^{-1} \log\log n))$ edges and $\beta = O({{\log\log n} \over \eps})^{\log\log n}$. 
This result of \cite{Pet09} is not efficient in the sense considered in this paper, i.e., no distributed or parallel implementations of it are known, and also no efficient (that is, roughly $O(|E|)$-time) centralized algorithm computing it is known. Also,  this result does not extend (\cite{Pettie-privcomm}) to a general tradeoff between $\beta$ and the number of edges.

Improving upon previous results by \cite{E05,EZ06}, Pettie \cite{P10} also devised an efficient distributed-CONGEST algorithm, that for a parameter $\rho > 0$, constructs in $\tO(n^\rho)$ rounds a $(1+ \eps,\beta)$-spanner with $O(n^{1+1/\kappa}  (\eps^{-1} \log\kappa)^\phi)$ edges and $\beta =O
\left({{\log \kappa + 1/\rho} \over \eps}\right)^{\log_\phi \kappa + 1/\rho}$, for $\phi = {{1  +\sqrt{5}} \over 2}$ being the golden ratio.
\footnote{
In the range of $\kappa =  o ({{\log n} \over {\log\log n}})$, the result of \cite{P10} is incomparable with \cite{EP04}, as spanners of \cite{EP04} provide  smaller $\beta$, while spanners of \cite{P10} are slightly sparser.}
Independently and simultaneously to our work, \cite{ABP17} showed that there exist $(1+ \eps,\beta)$-spanners with $O((\eps^{-1}\log\kappa)^h\cdot \log\kappa\cdot n^{1+1/\kappa})$ edges and $\beta =O\left({{\log \kappa } \over \eps}\right)^{\log \kappa-2}$, where $h=\frac{(3/4)\kappa -1-\log\kappa}{\kappa}<3/4$. This spanner has improved dependence on $\eps$ in the number of edges (at the cost of worse dependence on $\kappa$).

In this paper we improve all of the tradeoffs \cite{EP04,P10} in the entire range of parameters.
Specifically, for any $\eps > 0$, $\rho>0$ and $\kappa = 1,2,\ldots, {{\log n}\over{\log(1/\epsilon)+\log\log\log n}}$, our distributed-CONGEST algorithm constructs in $\tO(n^\rho)$ rounds a $(1+\eps,\beta)$-spanner with $O(n^{1+1/\kappa})$ edges and
$$\beta \le O\left({{\log \kappa  +1/\rho} \over \eps}\right)^{\log \kappa +1/\rho}~.$$
Our algorithm also admits efficient implementations in the streaming and standard models of computation, see Section \ref{sec:sp_central}.
Our spanners are sparser and have polynomially smaller $\beta$ than the previous best efficient constructions. They are even sparser than the state-of-the-art existential ones (with essentially the same $\beta$), with the following exceptions: whenever $\eps<1/\log^3\log n$ our result and that of \cite{ABP17} are incomparable,\footnote{The $i$-iterated logarithm is defined by $\log^{(i+1)} n = \log (\log^{(i)} n)$, for all $i \ge 0$, and $\log^{(0)} n =n$.} and
the spanner from \cite{Pet09} that has $O(n \log^{(4)} n)$ edges, while ours never gets sparser than $O(n\log\log n)/\eps$.
In the complementary range, $\eps > 1/\log^3\log n$, our result is strictly stronger than that of \cite{ABP17}.

Moreover, a variant of our algorithm efficiently constructs very sparse $(1 + \eps,\beta)$-emulators. In particular, we can obtain a {\em linear-size} $(1 + \eps,({{\log\log n} \over \eps})^{\log\log n})$-emulator. (We stress that the number of edges does not depend even on $\epsilon$.) All previous constructions of $(1+\eps,\beta)$-spanners or emulators  employ a superlinear number of edges, for all choices of parameters.

We use our new algorithms for constructing near-additive spanners and emulators to improve approximate shortest paths' algorithms, in the centralized and streaming models of computation. One notable result in this context is a streaming algorithm that for any constant $\eps > 0$ and any subset $S \subseteq V$ with $|S| = n^{\Omega(1)}$, computes $(1+\eps)$-approximate shortest paths for $S \times V$ within $O(|S|)$ passes over the stream, using $O(n^{1+\eps})$ space. See Section \ref{sec:appls} for more details, and additional applications of our spanners.

\subsection{Technical Overview}

Linial and Saks \cite{LS93} were the first to employ exponential random variables to build {\em network decompositions}, i.e., partitions of graphs into clusters of small diameter, which possess some useful properties. This technique was found useful for constructing padded partitions, hierarchically-separated trees,  low-stretch spanning trees \cite{Bar96,B98,B04,EEST05,ABN06,AN12} and spanners \cite{Coh93,BS03,E07}. In \cite{LS93} every vertex $v$ tosses a random variable $r_v$ from an exponential distribution, and broadcasts to all vertices within distance $r_v$ from $v$.
Every vertex $v$ joins the cluster of a vertex $u$ with largest identity number, whose broadcast $u$ heard.

Blelloch \etal \cite{BGKMPT11} introduced a variant of this technique in which, roughly speaking,  every vertex $v$ starts to broadcast at time $-r_v$, and broadcasts indefinitely.
A vertex $x$ joins the cluster centered at a vertex $v$, whose broadcast reaches $x$ first.
They called the resulting partition ``exponential start time clustering",  and it was demonstrated in \cite{BGKMPT11,MPX13,EN16b}  that this approach leads to very efficient distributed and parallel algorithms for constructing padded partitions and network decompositions.

Miller \etal \cite{MPVX15} used this approach to devise an efficient parallel and distributed-CONGEST $O(k)$-time algorithm for  constructing $O(k)$-spanner with $O(n^{1+1/k})$ edges. Specifically, they build the exponential time clustering, add the spanning trees of the clusters into the spanner, and then every vertex $x$ adds into the spanner one edge $(x,y)$ connecting $x$ to every adjacent cluster $C_y$, $y \in C_y$.

The main property of the partition exploited by \cite{MPVX15} in the analysis of their algorithm is that any unit-radius ball in the input graph $G$ intersects just $O(n^{2/k})$ clusters, in expectation. Note also that their algorithm is doomed to use at least $n^{1+1/k} + (n-1)$ edges, because it starts with inserting the spanning trees of all clusters (amounting to up to $n-1$ edges), and then inserts the $O(n^{1+2/k})$ edges crossing between different clusters into the spanner.
To get $O(n^{1+1/k})$ edges, one rescales $k' = 2k$.

In our algorithm we do not explicitly construct the exponential start time clustering. Rather we run the procedure that builds it, but every vertex $x$ connects not just to the neighbor $y$ through which $x$ received its first broadcast message at time, say, $t_y$, but also to all neighbors $z$ whose broadcast $x$ received witin time interval $[t_y,t_y+1]$. We show that, in expectation, $x$ connects to $n^{1/k}$ neighbors altogether, and not just to that many adjacent clusters.
As a result we obtain both a sparser spanner, a smaller stretch, and a smaller running time. The stretch and running time are smaller roughly by a factor of 2 than in \cite{MPVX15}, because we do not need to consider unit balls, that have diameter 2. Rather we tackle individual edges (of length  1).

In the context of {\em weighted} graphs, \cite{MPVX15} showed how their efficient algorithm for constructing sparse $(4k-2)$-spanners for unweighted graphs can be converted into an efficient algorithm that constructs $(16k-8)$-spanners with $O(n^{1+1/k} \log k)$ edges for weighted graphs. By using their scheme naively on top of our algorithm for unweighted graphs, one gets an efficient algorithm for computing $(4k-2)(1+\eps)$-spanners with $O(n^{1+1/k} \cdot (\log k)/\eps)$ edges. Roughly speaking, the overhead of 2 in the stretch is because in the analysis of \cite{MPVX15}, every vertex contributes expected $O(n^{1/k})$ edges to the spanner on each of roughly $O(n^{1/k})$ phases of the algorithm in which it participates. By employing a more delicate probabilistic argument, we argue that in fact, the expected total contribution of every vertex in {\em all phases altogether} is $O(n^{1/k})$, rather than $O(n^{2/k})$. This enables us to eliminate another factor of 2 from the stretch bound. See Section \ref{sec:wt} for details.

Our constructions of $(1+\eps,\beta)$-spanners and emulators follow the \cite{EP04} {\em superclustering and interconnection} approach.
 One starts with a base partition $\cP_0$. In \cite{EP04}  this was the partition of \cite{Awe85,PS89,AP92}, obtained via region-growing technique. Then every cluster $C \in \cP_0$ that has ``many" unclustered clusters of $\cP_0$ ``nearby", creates a supercluster around it. The ``many" and the ``nearby" are determined by degree threshold $\deg_0$ and distance threshold $\delta_0$, respectively. Once the superclustering phase is over, the remaining unclustered clusters enter an interconnection phase, i.e., every pair of participating nearby clusters gets interconnected by a shortest path in the spanner.  This completes one iteration of the process. The resulted superclustering $\cP_1$ is the input for the next iteration of this process, which runs with different, carefully chosen thresholds $\deg_1$ and $\delta_1$. Such iterations continue until only very few clusters survive. The latter are interconnected without further superclustering.

One bottleneck in devising efficient distributed algorithm based on this approach is the base partition. Known algorithms for constructing a region-growing partition of \cite{Awe85} require almost linear distributed time \cite{DMZ06}. We demonstrate that one can bypass it completely, and start from the base partition $\cP_0 = \{\{v\} \mid v \in V\}$. This requires some modfication of the algorithm, and a more nuanced analysis. In addition, we show that the superclustering and interconnection steps themselves can be implemented efficiently. This part of the algorithm is based on our recent work on hopsets \cite{EN16c}, where we showed that \cite{EP04} approach is extremely useful in that context as well, and that it can be made efficient.

\subsection{Related Work}

Efficient algorithms for constructing $(1+\eps,\beta)$-spanners were also devised in \cite{E05,EZ06,TZ06,Pet09,P10}.  These  algorithms are based, however, on different approaches than that of the current paper. The latter is based on \cite{EP04}. Specifically, the approach of \cite{E05,EZ06} is based on \cite{Coh93,C00} construction of pairwise covers and hopsets, i.e., the algorithm works top-down. It recurses in small clusters, and eliminates large ones. The approach of \cite{TZ06,Pet09,P10} is based on \cite{TZ01} collection of trees, used originally for distance oracles.

Streaming algorithms for constructing multiplicative spanners were given in \cite{FK+05,E07,B08}, and near-additive spanners in \cite{E05,EZ06}.
Spanners and emulators with sublinear error were given in \cite{TZ06,Pet09}. Spanners with purely additive error and lower bounds concerning them were given in \cite{ACIM99,EP04,BCE05,BKMP10,C13,W06,BW15,AB16}.

\subsection{Organization}

In Section \ref{sec:spanner} we present our algorithm for constructing multiplicative spanners and its analysis.
In Section \ref{sec:wt} we use this algorithm to provide improved spanners for weighted graphs as well.
Our near-additive spanners and emulators are presented in Section \ref{sec:sp_central}. 

\section{Sparse Multiplicative Spanners and Skeletons}\label{sec:spanner}

Let $G=(V,E)$ be a graph on $n$ vertices, and let $k\ge 1$ be an integer.
Let $c>3$ be a parameter governing the success probability, and set $\beta=\ln(cn)/k$. Recall the exponential distribution with parameter $\beta$, denoted $\ex(\beta)$, which has density
\[
f(x)=\left\{\begin{array}{ccc} \beta\cdot e^{-\beta x} & x\ge 0\\
0 & \text{otherwise.} \end{array} \right.
\]
\paragraph{Construction.}
Each vertex $u\in V$ samples a value $r_u$ from $\ex(\beta)$, and broadcasts it to all vertices within distance $k$. Each vertex $x$ that received a message originated at $u$, stores $m_u(x)=r_u-d_G(x,u)$, and also a neighbor $p_u(x)$ that lies on a shortest path from $x$ to $u$ (this neighbor sent $x$ the message from $u$, breaking ties arbitrarily if there is more than one). Let $m(x)=\max_{u\in V}\{m_u(x)\}$, then for every $x\in V$ we add to the spanner $H$ the set of edges
\[
C(x)=\left\{(x,p_u(x))~:~ m_u(x)\ge m(x)-1\right\}~.
\]

The following lemma is implicit in \cite{MPVX15}. We provide a proof for completeness.
\begin{lemma}[\cite{MPVX15}]\label{lem:MPVX}
Let $d_1\le\ldots\le d_n$ be arbitrary values and let $\delta_1,\dots,\delta_n$ be independent random variables sampled from $\ex(\beta)$. Define the random variables $M=\max_i\{\delta_i-d_i\}$ and $I=\{i~:~\delta_i-d_i\ge M-1\}$. Then for any $1\le t\le n$,
\[
\Pr[|I|\ge t]= (1-e^{-\beta})^{t-1}~.
\]
\end{lemma}
\begin{proof}
Denote by $X^{(t)}$ the random variable which is the $t$-th largest among $\{\delta_i-d_i\}$. Then for any value $a\in \R$, if we condition on $X^{(t)}=a$, then the event $|I|\ge t$ is exactly the event that all the remaining $t-1$ values $X^{(1)},\dots,X^{(t-1)}$ are at least $a$ and at most $a+1$. Using the memoryless property of the exponential distribution and the independence of the $\{\delta_i\}$, we have that
\[
\Pr[|I|\ge t\mid X^{(t)}=a]=(1-e^{-\beta})^{t-1}~.
\]
Since this bound does not depend on the value of $a$, applying the law of total probability we conclude that
\[
\Pr[|I|\ge t] = (1-e^{-\beta})^{t-1}~.
\]
\end{proof}
Using this lemma, we can bound the expected size of the spanner.
\begin{lemma}\label{lem:size}
The expected size of $H$ is at most $(cn)^{1/k}\cdot n$.
\end{lemma}
\begin{proof}
Fix any $x\in V$, and we analyze $\E[|C(x)|]$. Note that the event $|C(x)|\ge t$ happens when there are at least $t$ shifted random variables $r_u-d_G(u,x)$ that are within 1 of the maximum. By Lemma \ref{lem:MPVX} this happens with probability at most $(1-e^{-\beta})^{t-1}$ (we remark that if $x$ did not hear at least $t$ messages, then trivially $\Pr[|C(x)|\ge t]=0$). We conclude that
\begin{eqnarray*}
\E[|C(x)|]&=&\sum_{t=1}^n\Pr[|C(x)|\ge t]
\le\sum_{t=0}^\infty(1-e^{-\beta})^t=e^{\beta} = (cn)^{1/k}~,
\end{eqnarray*}
and the lemma follows by linearity of expectation.
\end{proof}

We now argue about the stretch of the spanner.

\begin{claim}\label{claim:radius}
With probability at least $1-1/c$, it holds that $r_u< k$ for all $u\in V$.
\end{claim}
\begin{proof}
For any $u\in V$, $\Pr[r_u\ge k] = e^{-\beta k} = 1/(cn)$. By the union bound, $\Pr[\exists u,~r_u\ge k]\le 1/c$.
\end{proof}
Assume for now that the event of Claim \ref{claim:radius} holds, i.e., that $r_u< k$ for all $u\in V$.
\begin{corollary}\label{cor:max}
For any $x\in V$, if $u\in V$ is the vertex maximizing $m_u(x)$, then $d_G(u,x)<k$.
\end{corollary}
\begin{proof}
First note that $m(x)\ge m_x(x)\ge 0$, and using Claim \ref{claim:radius} we have $r_u<k$. So $0\le m(x)=m_u(x)=r_u-d_G(u,x)<k-d_G(u,x)$.
\end{proof}
\begin{claim}\label{claim:path}
For any $u,x\in V$, if $x$ adds an edge to $p_u(x)$, then there is a shortest path $P$ between $u$ and $x$ that is fully contained in the spanner $H$.
\end{claim}
\begin{proof}
We prove by induction on $d_G(u,x)$. In the base case $d_G(x,u)=1$, then $p_u(x)=u$, so $(x,u)$ is in the spanner. Assume that every vertex $y\in V$ with $d_G(u,y)=t-1$ which added an edge to $p_u(y)$ has a shortest path to $u$ in $H$, and we prove for $x$ that has $d_G(u,x)=t$. We know that $x$ added an edge to $y=p_u(x)$, which lies on a shortest path to $u$, and thus satisfies $d_G(u,y)=t-1$. It remains to show that this $y$ added an edge to $p_u(y)$. First we claim that
\begin{equation}\label{eq:1}
m(y)\le m(x)+1~.
\end{equation}
Seeking contradiction, assume that \eqref{eq:1} does not hold, and let $v\in V$ be the vertex maximizing $m_v(y)$. By Corollary \ref{cor:max} we have $d_G(v,y)<k$, and thus $d_G(v,x)\le k$. Hence  $x$ will hear the message of $v$. This means that $m_v(x)\ge m_v(y)-1=m(y)-1>m(x)$, which is a contradiction. This establishes \eqref{eq:1}. Now, since $x$ added an edge to $y=p_u(x)$, by construction
\begin{equation}\label{eq:2}
m_u(x)\ge m(x)-1~.
\end{equation}
We conclude that
\[
m_u(y)= m_u(x)+1\stackrel{\eqref{eq:2}}{\ge} m(x)-1+1\stackrel{\eqref{eq:1}}{\ge} m(y)-1~,
\]
so $y$ indeed adds an edge to $p_u(y)$, and by the induction hypothesis we are done.
\end{proof}

\begin{lemma}\label{lem:stretch}
The spanner $H$ has stretch at most $2k-1$.
\end{lemma}
\begin{proof}
Since $H$ is a subgraph of $G$, it suffices to prove for any $(x,y)\in E$, that $d_H(x,y)\le 2k-1$. Let $u$ be the vertex maximizing $m(x)=m_u(x)$, and w.l.o.g assume $m(x)\ge m(y)$. By Corollary \ref{cor:max} we have $d_G(u,x)\le k-1$, so $d_G(u,y)\le k$, and $y$ heard the message of $u$ (which was sent to distance $k$). This implies that $m_u(y)\ge m_u(x)-1=m(x)-1\ge m(y)-1$, so $y$ adds the edge $(y,p_u(y))$ to the spanner. By applying Claim \ref{claim:path} on $x$ and $y$, we see that both have shortest paths to $u$ that are fully contained in $H$. Since $d_G(x,u)\le k-1$ and $d_G(y,u)\le k$, these two paths provide stretch $2k-1$ between $x,y$.
\end{proof}

\subsection{Main Theorem}

We now state our main theorem, from which we will derive several interesting corollaries in various settings.

\begin{theorem}\label{thm:MPX-spanner}
For any unweighted graph $G=(V,E)$ on $n$ vertices, any integer $k\ge 1$, $c>3$ and $\delta>0$, there is a randomized algorithm that with probability at least $(1-1/c)\cdot\delta/(1+\delta)$ computes a spanner with stretch $2k-1$ and number of edges at most
\[
(1+\delta)\cdot\frac{(cn)^{1+1/k}}{c-1}-\delta(n-1)~.
\]
\end{theorem}
\begin{proof}
Let ${\cal Z}$ be the event that $\{\forall u\in V,~ r_u<k\}$. By Claim \ref{claim:radius} we have $\Pr[{\cal Z}]\ge 1-1/c$. Note that conditioning on ${\cal Z}$, by Lemma \ref{lem:stretch} the algorithm produces a spanner $H=(V,E')$ with stretch $2k-1$. In particular, it must have at least $n-1$ edges. Let $X$ be the random variable $|E'|-(n-1)$, which conditioned on ${\cal Z}$ takes only nonnegative values. By Lemma \ref{lem:size} we have $\E[X]=(cn)^{1/k}\cdot n-(n-1)$.
We now argue that conditioning on ${\cal Z}$ will not affect this expectation by much. Indeed, by the law of total probability, for any $t$, $\Pr[X=t]\ge \Pr[X=t\mid{\cal Z}]\cdot\Pr[{\cal Z}]$. Thus

\begin{equation}\label{eq:expect}
\E[X\mid{\cal Z}]\le \frac{\E[X]}{\Pr[{\cal Z}]}\le\frac{c}{c-1}\cdot\left[(cn)^{1/k}\cdot n-(n-1)\right]~.
\end{equation}
By Markov inequality,
\[
\Pr\left[X\ge (1+\delta)\E[X\mid{\cal Z}]\mid{\cal Z}\right]\le \frac{1}{1+\delta}~.
\]

We conclude that

\begin{eqnarray*}
\Pr\left[\left(X< (1+\delta)\E[X\mid{\cal Z}]\right)\wedge{\cal Z}\right]= \Pr\left[X<(1+\delta) \E[X\mid{\cal Z}]\mid{\cal Z}\right]\cdot\Pr[{\cal Z}]\ge \left(1-\frac{1}{c}\right)\cdot\frac{\delta}{1+\delta}~.
\end{eqnarray*}
If this indeed happens, then
\begin{eqnarray*}
|E'| &=& X+n-1\\
&\stackrel{\eqref{eq:expect}}{\le}&\!\!\!\! (1+\delta)\cdot\frac{c}{c-1}\cdot\left[(cn)^{1/k}\cdot n-(n-1)\right]+n-1\\
&=&\!\!\!\! (1+\delta)\cdot\frac{(cn)^{1+1/k}-(n-1)}{c-1}-\delta (n-1)~.
\end{eqnarray*}

\end{proof}


\subsubsection{Implementation Details}\label{sec:implement}

\paragraph{Distributed Model.}
It is straightforward to implement the algorithm in the LOCAL model of computation, it will take $k$ rounds to execute it -- in each round, every vertex sends to its neighbors all the messages it received so far. We claim that the algorithm can be implemented even when bandwidth is limited, i.e., in the CONGEST model. This will require a small variation: in each round, every vertex $v\in V$ will send to all its neighbors the message $(r_u,d_G(u,v))$ for the vertex $u$ that currently maximizes $m_u(v)=r_u-d_G(u,v)$. We note that omitting all the other messages will not affect the algorithm, since if one such message would cause some neighbor of $v$ to add an edge to $v$, then the message about $u$ will suffice, as the latter has the largest $m_u(v)$ value. 
(Also recall that all vertices start their broadcast simultaneously, and do so for $k$ rounds, so any omitted message could not have been sent to further distance than the message from $u$, which implies dropping it will have no effects on farther vertices as well.)

\paragraph{PRAM Model.}
In the parallel model of computation, we can use a variant of the construction that appeared in \cite{MPX13,MPVX15}. Roughly speaking, vertex $u$ will start its broadcast at time $k-\lceil r_u\rceil$, and every vertex $x$ will send only the first message that arrives to it (which realizes $m(x)$). As argued in \cite{MPVX15}, the algorithm can be implemented in $O(k\log^*n)$ depth and $O(|E|)$ work.

\paragraph{Standard Centralized Model.}
Note that in the standard centralized model of computation, the running time is at most the work of the PRAM algorithm, which is $O(m)$. By taking constant $c$ and $\delta$, and repeating the algorithm until the first success (we can easily check the number of edges of the spanner and that all $r_u<k$), we get a spanner with stretch $2k-1$ and $O(n^{1+1/k})$ edges in expected time $O(|E|)$.

\subsection{Implications of Theorem \ref{thm:MPX-spanner}}

\subsubsection{Standard Centralized Model and PRAM}

The currently sparsest spanners which can be constructed in linear time are those of Halperin and Zwick \cite{HZ96}. They provide for any $k\ge 1$, a deterministic algorithm running in $O(m)$ time, that produces a spanner with $n^{1+1/k}+n$ edges. We can improve their result for a wide range of $k$, albeit with a randomized algorithm. First we show a near-linear time algorithm (which can be also executed in parallel), that provides a spanner sparser than Halperin and Zwick in the range $k\ge 2\ln n/\ln\ln n$.\footnote{In fact, the factor 2 can be replaced by any $1+\epsilon$ for constant $\epsilon>0$.}

\begin{corollary}\label{cor:sequential}
For any unweighted graph $G=(V,E)$ on $n$ vertices and $m$ edges, and any integer $k\ge 2$, there is a randomized algorithm, that with high probability\footnote{By high probability we mean probability at least $1-n^{-C}$, for any desired constant $C$.} computes a spanner for $G$ with stretch $2k-1$ and $n^{1+1/k}\cdot\left(1+\frac{O(\ln k)}{k}\right)$ edges. The algorithm has $O(k^2\ln n\ln^*n)$ depth and the running time (or work) is $\tilde{O}(k|E|)$.
\end{corollary}
\begin{proof}
Apply Theorem \ref{thm:MPX-spanner} with parameters $c=k$ and $\delta=1/k$.  So with probability at least $\frac{k-1}{k}\cdot\frac{1}{k+1}\ge \frac{1}{3k}$ we obtain a spanner whose number of edges is at most
\begin{equation}\label{eq:qqw}
(1+1/k)\cdot \frac{(kn)^{1+1/k}}{k-1}\le n^{1+1/k}\cdot\left(1+\frac{O(\ln k)}{k}\right)~.
\end{equation}

Run the algorithm $C\cdot k\ln n$ times for some constant $C$. We noted in Section \ref{sec:implement} that each run takes $O(k\ln^*n)$ depth and $O(|E|)$ work, so the time bounds are as promised. Now, with probability at least $1 - (1-1/(3k))^{C\cdot k\ln n} \ge 1 - n^{-C/3}$,  we achieved a spanner with number of edges as in \eqref{eq:qqw} in one of the executions.
\end{proof}

\begin{remark}
Whenever $k\ge 2\ln n/\ln\ln n$, we have $n^{1/k}\le \sqrt{\ln n}$, so the number of edges in Corollary \ref{cor:sequential} is $n^{1+1/k}+o(n)$, and the running time is $\tilde{O}(k|E|)$.
\end{remark}

\subsubsection{Distributed Model}
In a distributed setting we have the following result.
\begin{corollary}\label{cor:dist-spanner}
For any unweighted graph $G=(V,E)$ on $n$ vertices, any $k\ge 1$ and $0<\epsilon<1$, there is a randomized distributed algorithm that with probability at least $1-\epsilon$ computes a spanner with stretch $2k-1$ and $O(n/\epsilon)^{1+1/k}$ edges, within $k$ rounds.
\end{corollary}
\begin{proof}
Apply Theorem \ref{thm:MPX-spanner} with $c=3/\epsilon$ and $\delta =2/\epsilon$, so the success probability is at least
\[
(1-\epsilon/3)\cdot(1-\epsilon/(\epsilon+2))>1-\epsilon~.
\]
With these parameters, by Theorem \ref{thm:MPX-spanner}, the number of edges in spanner will be bounded by $O(n/\epsilon)^{1+1/k}$.
\end{proof}

\subsubsection{Ultra-Sparse Spanners and Skeletons}

We now show that in the regime $k\ge\ln n$, our algorithm (that succeeds with probability close to 1) provides a spanner whose number of edges is very close to $n$ (as a function of $k$ and the success probability). This will hold in all computational models we considered. We note that for the centralized and PRAM models, Corollary \ref{cor:sequential} gives high probability with roughly the same sparsity, albeit with larger depth and work.

\begin{corollary}\label{cor:sequential-linear}
For any unweighted graph $G=(V,E)$ on $n$ vertices, and any integer $k\ge\ln n$ and parameter $2/k<\epsilon<1$, there is a randomized algorithm, that with probability at least $1-\epsilon$ computes a spanner for $G$ with stretch $2k-1$ and $n\cdot\left(1+\frac{O(\ln n)}{\epsilon\cdot k}\right)$ edges.
The number of rounds in distributed model is $k$, in PRAM it is $O(k \log^* n)$ depth and $O(|E|)$ work, and in the centralized model it is $O(|E|)$ time.
\end{corollary}
\begin{proof}
Apply Theorem \ref{thm:MPX-spanner} with parameters $c=k$ and $\delta=2/\epsilon$, so with probability at least $\frac{k-1}{k}\cdot\frac{2/\epsilon}{2/\epsilon+1}\ge 1-\epsilon$ we obtain a $(2k-1)$-spanner. In the regime $k\ge\ln n$ we have $(kn)^{1/k}\le e^{(2\ln n)/k}\le 1+O(\ln n)/k$, so the number of edges is at most
\begin{equation}\label{eq:qqqw}
(1+2/\epsilon)\cdot \frac{(kn)^{1+1/k}}{k-1}-2(n-1)/\epsilon\le n\cdot\left(1+\frac{O(\ln n)}{\epsilon\cdot k}\right) ~.
\end{equation}
\end{proof}

\begin{remark}
The spanner of Corollary \ref{cor:sequential-linear} can be used as a skeleton. E.g., one can take  $\epsilon=o(1)$ and $k=\tilde{O}(\log n)$, to obtain with probability $1-o(1)$, a skeleton with $n(1+o(1))$ edges, which is computed in $\tilde{O}(\log n)$ rounds.
\end{remark}

\subsection{Weighted Graphs}
\label{sec:wt}


Miller et al. \cite{MPVX15} used their efficient algorithm for constructing $O(k)$-spanners with $O(n^{1+1/k})$ edges for unweighted graphs, to provide an efficient algorithm for constructing $O(k)$-spanners with $O(n^{1+1/k} \log k)$ edges for weighted graphs. In this section we argue that their scheme can be used to convert our algorithm for constructing $(2k-1)$-spanners with $O(n^{1+1/k})$ edges for unweighted graphs (Theorem \ref{thm:MPX-spanner}; see also Section \ref{sec:implement}) into an efficient algorithm for constructing $(2k-1)(1+\eps)$-spanners with $O(n^{1+1/k} \cdot {{\log k} \over \eps})$ edges for weighted graphs.

The scheme of \cite{MPVX15} works in the following way.  It partitions all edges of $G = (V,E)$ into $\lceil \log_{1+\eps} \omega_{max} \rceil = \lambda$ categories $E_t = \{e \in E \mid \omega(e) \in [(1+\eps)^{t-1},(1+\eps)^t)\}$, $ t = 1,2,\ldots,\lambda$. (We assume that the minimum weight is 1, and the maximum weight is $\omega_{max}$. The last category $E_\lambda$ should also contain edges of weight exactly $\omega_{max}$.) Now one defines $\ell = \lceil \log_{1+ \eps} k^c \rceil$, for a sufficiently large constant $c$, graphs $G_j = (V,\hE_j)$, $j = 0,1,\ldots, \ell-1$, $\hE_j = \bigcup \{E_t \mid t \equiv j (\mod \ell)\}$.

Observe that the edge weights in (each) $G_j$ are {\em well-separated}, i.e., $\hE_j$ is a disjoint union of at most  $q = \lceil \lambda/\ell \rceil$ edge sets $E^{(1)},\ldots,E^{(q)}$, such that the edge weights within each set are within a factor of $1+ \eps$ from one another. Moreover, if edge weights in $E^{(i)}$, $i = 1,2,\ldots,q$, are in the range $[w^{(i)},(1+\eps)w^{(i)})$, then we have $w^{(i)} = w^{(1)} (k^c)^{i-1}$, for every $i = 1,2,\ldots,q$.

For each graph $G_j$, the scheme of \cite{MPVX15} constructs an $O(k)$-spanner with $O(n^{1+1/k})$ edges. It then takes a union of $\ell = O({{\log k}  \over \eps})$ such spanners as the ultimate $O(k)$-spanner of the original graphs. (In fact, \cite{MPVX15} used specifically $\eps=1$.)  We will next outline the way in which \cite{MPVX15} construct $O(k)$-spanner with $O(n^{1+1/k})$ edges for each $G_j$, and show how to modify it to provide a $(2k-1)(1+\eps)$-spanner.

The scheme starts with running a routine of \cite{MPVX15} that constructs an $O(k)$-spanner $H^{(1)}$ with $O(n^{1+1/k})$ edges for the {\em unweighted} graph $(V^{(1)},E^{(1)})$, $V^{(1)} = V$, and constructing the exponential start time partition $\cP^{(1)}$ for it. It then contracts each of the clusters of
$\cP^{(1)}$
(which have unweighted radii at most $k-1$) into single vertices of $V^{(2)}$, and runs the unweighted spanner routine on $(V^{(2)},E^{(2)})$. As a result, it constructs an $O(k)$-spanner $H^{(2)}$ and a partition $\cP^{(2)}$ of $V^{(2)}$, contracts all clusters of $\cP^{(2)}$ to get $V^{(3)}$, etc. The final spanner returned by the scheme is $H = \bigcup_{i=1}^q H^{(i)}$.

The scheme guarantees stretch $(1+\eps)(1 + O(k^{-(c-1)}))O(k)$, because the blackbox routine for unweighted graphs provides stretch $O(k)$ for each category of weights, but the weights are uniform only up to a factor of $1 + \eps$. Also, the factor of $1 + O(k^{-(c-1)})$ appears, because one contracts clusters of unweighted diameter $O(k)$ of lower scales, on which all edge weights are a factor of roughly $k^{-c}$ smaller than the edge weights on the current scale.

In the analysis of $|H|$, \cite{MPVX15} show that every vertex $u$ is active (i.e., non-isolated vertex which is not yet contracted into a larger super-vertex) for expected $O(n^{1/k})$ phases, and when it is active, it contributes expected $O(n^{1/k})$ edges to the spanner of the current phase. Hence the overall size of the spanner is $O(n^{1+2/k})$, and by rescaling $k' = 2k$, they ultimately get their result.
(See the proof of Theorem 3.3 in \cite{MPVX15} for full details of this proof. We have sketched it for the sake of completeness.)

While the stretch analysis of \cite{MPVX15} is sufficiently precise for our purposes, this is not the case with the size analysis. Indeed, even when one plugs in stretch $2k-1$ 
of our unweighted spanner routine instead of stretch $O(k)$ of their routine, still one obtains a $(2k-1)(1+\eps)^2(1 +  k^{-(c-1)})$-spanner with $O(n^{1+2/k})$ edges, i.e., a $(4k-2)(1+O(\eps))$-spanner with $O(n^{1+1/k})$ edges (for each $G_j$).

In what follows we refine their size analysis, and show that, in fact, every vertex $u$ contributes expected $O(n^{1/k})$ edges in {\em all phases} of the algorithm {\em altogether} (for a single graph $G_j$ with well-separated edge weights). Denote by $r_u^{(i)}$ the radius that $u$ tosses from $\ex(\beta)$ in the $i$th phase, $i = 1,2,\ldots,q$, assuming that it is active on that phase. We say that a vertex $v$ (which is active on phase $i$) is a {\em candidate} vertex of phase $i$ if its broadcast message reaches $u$ no later than within one time unit after the time $-r_u^{(i)}$, i.e., $-r_u^{(i)} + 1 \ge -r_v^{(i)} + d(v,u)$, where $d(v,u)$ is the unweighted distance between $v$ and $u$ in the graph on which the unweighted spanner routine is invoked on phase $i$.

Let $X^{(i)}$ denote the random variable counting the number of such candidate vertices on phase $i$, for $i = 1,2,\ldots,q$. (Recall that these are the vertices which might cause $u$ to add an edge to the spanner.)
Let $j$ denote the random variable which is the phase in which $u$ was contracted (and $j=q$ if there is no such phase). That is, $j$ indicates the level in which the broadcast of some candidate vertex $v$ has $ -r_v^{(i)} + d(v,u)<-r_u^{(i)}$.
Denote also by $\hX^{(i)}$, $i = 1,2,\ldots,q$, the total number of candidates $u$ sees between the beginning of phase $i$, and up until phase $j$, where a candidate vertex $v$ reaches $u$ before time $-r_u^{(j)}$. On that phase $j$, $\hX^{(i)}$ counts the number of candidates with index not larger than that of the candidate vertex $v$ (assume every vertex has an arbitrary distinct index in $\{1,\ldots n\}$).
Note that for $i > j$, by definition, $\hX^{(i)} = 0$.
Also, in particular, $\hX = \hX^{(1)}$ is at least the total contribution of edges $u$ adds to the spanner, except for up to  expected $O(n^{1/k})$ edges that it might contribute on phase $j$ (as shown in Lemma~\ref{lem:size}).

We next argue that for any $t = 0,1,2,\ldots$,
\begin{equation}
\label{eq:hX}
\Pr(\hX > t) \le (1 - e^{-\beta})^t~.
\end{equation}
This implies that
$$\E(\hX) ~=~ \sum_{t=0}^\infty \Pr(\hX > t) ~\le \sum_{t=0}^\infty (1 - e^{-\beta})^t ~=~ e^\beta ~=~ O(n^{1/k})~.$$
First, note that
$$\Pr(\hX > t) ~=~ \sum_{t^{(1)} = 0}^\infty \Pr(X^{(1)} = t^{(1)}) \cdot \Pr(\hX > t \mid X^{(1)} = t^{(1)})~.$$
For $t^{(1)} > t$, we have
$$\Pr(\hX > t \mid X^{(1)} = t^{(1)}) ~=~ (1 - e^{-\beta})^t~.$$
To justify this equation, note that the left-hand side is exactly the probability that on the first phase, all the $t$ first candidate vertices $v$ have
$-r_v^{(1)} + d(v,u) \ge - r_u^{(1)}$, conditioned on them being candidates, i.e., on $-r_v + d(v,u) \le -r_u^{(1)} + 1$. Since these are independent shifted exponential random variables, the equation follows from the memoryless property of the exponential distribution.

For $t^{(1)} \le t$, we have
\begin{eqnarray*}
\Pr(\hX > t \mid X^{(1)} = t^{(1)}) ~&=&~ (1 - e^{-\beta})^{t^{(1)}} \cdot \Pr(\hX^{(2)} > t - t^{(1)} \mid X^{(1)} = t^{(1)}, \hX^{(1)} \ge t^{(1)})\\
&=&(1 - e^{-\beta})^{t^{(1)}} \cdot \Pr(\hX^{(2)} > t - t^{(1)} \mid \hX^{(1)} \ge X^{(1)})~.
\end{eqnarray*}

(Again, $(1 - e^{-\beta})^{t^{(1)}}$ is the probability that no candidate vertex of the first phase reached $u$ before time $-r_u^{(1)}$, and so $u$ was not contracted away at this phase.)

We conduct an induction on the phase, where the induction base is the last phase $i = q$. On the last phase, for any $h$,
$$\Pr(\hX^{(q)} > h \mid \hX^{(1)} \ge X^{(1)},\hX^{(2)} \ge X^{(2)},\ldots,\hX^{(q-1)} \ge X^{(q-1)}) \le (1 - e^{-\beta})^h~,$$
because it is just the probability that none of the first $h$ candidates (that have $-r_v^{(q)} + d(v,u) \le - r_u^{(1)}+1$) reaches $u$ before time $-r_u^{(q)}$. (If there are fewer candidates or $u$ was contracted in a previous phase, then this probability is 0.)

Hence by the inductive hypothesis,
$$\Pr(\hX^{(2)} > t - t^{(1)} \mid \hX^{(1)} \ge X^{(1)}) ~\le~ (1 - e^{-\beta})^{t - t^{(1)}}~,$$
and so
$$\Pr(\hX^{(1)}  > t \mid X^{(1)} = t^{(1)}) ~\le ~ (1 - e^{-\beta})^t~,$$ for any $t^{(1)} \le t$ as well.
Hence
\begin{eqnarray*}
\Pr(\hX > t) &=& \sum_{t^{(1)} = 0}^\infty \Pr(X^{(1)} = t^{(1)}) \cdot \Pr(\hX > t \mid X^{(1)} = t^{(1)}) \\
&\le& \sum_{t^{(1)} = 0}^\infty \Pr(X^{(1)} = t^{(1)}) \cdot (1 - e^{-\beta})^t ~=~ (1 - e^{-\beta})^t~,
\end{eqnarray*}
as required.

Hence the expected contribution of every vertex is $O(n^{1/k})$, and the overall spanner size for each $G_j$ is, in expectation, $O(n^{1+1/k})$. The running time of the algorithm is expected to be $O(|E|)$, following the analysis of \cite{MPVX15}. Moreover, as in \cite{MPVX15}, our algorithm can be implemented in PRAM model, in $O(\log n \cdot \log^* n \cdot \log \Lambda )$ depth, and $O(|E|)$ work, where $\Lambda$ is the aspect ratio of the input graph. We summarize the result below.

\begin{theorem}
\label{thm:wt}
Given a weighted $n$-vertex graph $G$, and a pair of parameters $k \ge 1$, $0<\eps<1$, our algorithm computes a $(2k-1)(1+\eps)$-spanner of $G$ with $O(n^{1+1/k} \cdot (\log k)/\eps)$ edges, in expected $O(|E|)$ centralized  time, or in
$O(\log n \cdot \log^* n \cdot \log \Lambda )$ depth and $O(|E|)$ work.
\end{theorem}

The result of Theorem \ref{thm:wt} can be used in conjunction with the scheme of \cite{ES16} to devise an algorithm that computes $(2k-1)(1+\eps)$-spanners of size $O(n^{1+1/k}(\log k/\eps) 1/\eps)$, with {\em lightness} (i.e., weight of the spanner divided by the weight of the MST of the input graph) $O(k \cdot n^{1/k} (1/\eps)^{2+1/k})$, in expected time $O(|E| + \min\{n\log n,|E|\alpha(n)\})$, where $\alpha(\cdot)$ is an inverse-Ackermann function. This improves a result of \cite{ES16} that provides spanners with the same stretch and lightness, but with more edges (specifically, $O((k + (1/\eps)^{2+1/k}) n^{1+1/k})$, and using $O(k \cdot |E| + \min \{n \log n,|E| \alpha(n)\})$ time. Recently, consequently to our work, Alstrup et al. \cite{ADFSW17} further improved these bounds. We thus omit the details of our argument that provides the aforementioned bounds.

\section{An Efficient Centralized Construction of Nearly-Additive Spanners and Emulators}
\label{sec:sp_central}

\subsection{A Basic Variant of the Algorithm}
\label{sec:sp_basic}

In this section we present an algorithm for constructing $(1+\epsilon,\beta)$-spanners, which can be efficiently implemented in various settings.
We start with the centralized setting. In this setting we present two variants of our construction. The first variant presented in this section is somewhat simpler, while the second variant presented in the next section provides better bounds.

Let $G=(V,E)$ be an unweighted graph on $n$ vertices, and let $k\ge 1$, $\eps>0$ and $0<\rho<1/2$ be parameters.
Unlike the algorithm of \cite{EP04}, our algorithm does not employ sparse partitions of \cite{AP92}.
The algorithm initializes the spanner $H$ as an empty set, and proceeds in phases. It starts with setting $\chP_0 = \{\{v\} \mid v \in V\}$ to be the partition of $V$ into singleton clusters.
The partition $\chP_0$ is the input of phase 0 of our algorithm. More generally, $\chP_i$ is the input of phase $i$, for every index $i$ in a certain appropriate range, which we will specify in the sequel.

Throughout the algorithm, all clusters $C$ that we will construct will be centered at designated centers $r_C$. In particular, each singleton cluster $C= \{v\} \in \chP_0$ is centered at $v$. We define $\Rad(C) = \max \{d_{G(C)}(r_C,v) \mid v \in C\}$, and $\Rad(\chP_i) = \max_{C \in \chP_i} \{\Rad(C)\}$.

All phases of our algorithm except for the last one consist of two steps. Specifically, these are the {\em superclustering} and the {\em interconnection} steps.
The last phase contains only the interconnection step, and the superclustering step is skipped.
We also partition the phases
into two {\em stages}. The first stage consists of phases  $0,1, \ldots,i_0=\lfloor\log(\kappa\rho)\rfloor$, and the second stage consists of all the other phases $i_0+1,\ldots,i_1$ where $i_1=i_0+\left\lceil\frac{\kappa+1}{\kappa\rho}\right\rceil-2$, except for the last phase $\ell = i_1+1$. The last phase will be referred to as the {\em concluding phase}. 

Each phase $i$ accepts as input two parameters, the distance threshold parameter $\delta_i$, and the degree parameter $\deg_i$. The difference between stage 1 and 2 is that in stage 1 the degree parameter grows exponentially, while in stage 2 it is fixed. The distance threshold parameter
grows in the same steady rate (increases by a factor of $1/\epsilon$) all through the algorithm.

Next we describe the first stage of the algorithm. We start with describing its superclustering step.
We set $\deg_i = n^{2^i/\kappa}$, for all $i = 0,1,\ldots,i_0$.
Let $R_0=0$, and $\delta_i = (1/\epsilon)^i + 4 \cdot R_i$, where $R_i$ is determined by the following recursion: $R_{i+1}=\delta_i+R_i=(1/\epsilon)^i + 5 \cdot R_i$. We will show that the inequality $\Rad(\chP_i) \le R_i$ will hold for all $i$.

On phase $i$, each cluster $C \in \chP_i$ is sampled  i.a.r. with probability $1/\deg_i$. Let $\cS_i$ denote the set of sampled clusters.
We now conduct a BFS exploration to depth $\delta_i$ in $G$ rooted at the set $S_i = \bigcup_{C \in \cS_i} \{r_C\}$.
As a result,  a forest $F_i$ is constructed, rooted at vertices of $S_i$.
For a cluster center $r' = r_{C'}$ of a cluster $C' \in \chP_i \setminus \cS_i$ such that $r'$ is spanned by $F_i$,  let $r_C$ be the root of the forest tree of $F_i$ to which $r'$ belongs. (The vertex $r_C$ is by itself a cluster center of a cluster $C \in \cS_i$.)  The cluster $C'$ becomes now superclustered in a cluster $\hat{C}$ centered around the cluster $C$. (We also say that $C'$ is {\em  associated} with $C$. We will view association as a transitive relation, i.e., if $C'$ is associated with $C$ and $C''$ is associated with $C'$, we will think of $C''$ as associated with $C$ as well.)

The cluster center $r_C$ of $C$ becomes the new cluster center of $\hat{C}$, i.e., $r_{\hC} = r_C$. The vertex set of the new supercluster $\hat{C}$ is the union of the vertex set of $C$ with the  vertex sets of all clusters $C'$ which are superclustered into $\hC$. The edge set $T_{\hC}$ of the new cluster $\hC$ contains the BFS spanning trees of all these clusters, and, in addition, it contains shortest paths from the forest $F_i$ between $r_C$ and each $r_{C'}$ as above. $\chS_i$ is the set of superclusters created by this process. We set $\chP_{i+1} = \chS_i$.
All edges that belong to the edgeset of one of these superclusters are now added to the spanner $H$.

For each supercluster $\hC$, we write $\Rad(\hC) =  \Rad(T_{\hC},r_{\hC}).$
Observe that $\Rad(\chP_0) \le R_0=0$,
and
$\Rad(\chS_0) = \max \{\Rad(\hC) \mid \hC \in \chS_0\} \le \delta_0 + R_0 = R_1=1$.
More generally we have

\begin{eqnarray}
\label{eq:rad}
\Rad(\chS_i) = \max \{\Rad(\hC) \mid \hC \in \chS_i\}\le \delta_i + \Rad(\chP_i) \le \delta_i + R_i \le (1/\epsilon)^i + 5 R_i=R_{i+1}.
\end{eqnarray}

Denote by $\chU_i$ the set of clusters of $\chP_i$ which were not superclustered into clusters of $\chS_i$.
In the interconnection step for $i\ge 1$, every cluster center $r_C$ of a cluster $C \in \chU_i$ initiates a BFS exploration to depth $\half \delta_i$, i.e., half the depth of the exploration which took place in the superclustering step.
For each cluster center $r_{C'}$ for $C' \in \chP_i$ which is discovered by the exploration initiated in $r_C$, the shortest path between $r_C$ and $r_{C'}$ is inserted into the spanner $H$. The first phase $i=0$ is slightly different: the exploration depth is set to be 1, and we add an edge from $\{v\}\in \chU_0$ to all neighbors that are in $\chU_0$. This completes the description of the interconnection step.

\begin{lemma}
\label{lm:explorations}
For any vertex $v \in V$, the expected number of explorations that visit $v$ at the interconnection step of phase $i$ is at most $\deg_i$.
\end{lemma}
\begin{proof}
For $i\ge 1$, assume that there are $l$ clusters of $\chP_i$ whose centers are within distance $\delta_i/2$ from $v$.  If at least one of them is sampled to $\cS_i$, then no exploration will visit $v$ (since in the superclustering phase the sampled center will explore to distance $\delta_i$, and thus will supercluster all these centers). The probability that none of them is sampled is $(1-1/\deg_i)^l$, in which case we get that $l$ explorations visit $v$, so the expectation is $l\cdot(1-1/\deg_i)^l\le \deg_i$ (which holds for any $l$).

For $i=0$, we note that we add an edge touching $v$ iff none of its neighbors were sampled at phase 0 (as otherwise it would be clustered and thus not in $\chU_0$). The expected number of edges added is once again $l\cdot(1-1/\deg_i)^l\le \deg_i$ (here $l$ is the number of neighbors).
\end{proof}

We also note the following lemma for future use, its proof follows from a simple Chernoff bound.
\begin{lemma}
\label{lm:property}
For any constant $c > 1$, with probability at least $1 - 1/n^{c-1}$, for every vertex $v \in V$, at least one among the $\deg_i \cdot c \cdot \ln n$ closest cluster centers $r_{C'}$ with $C' \in \chP_i$ to $v$ is sampled, i.e., satisfies
$C' \in \cS_i$.
\end{lemma}

Observe that no vertex $v \in V$ is explored by more than $c \cdot \ln n \cdot \deg_i$ explorations, with probability at least $1 - n^{-(c-1)}$. Indeed, otherwise when $i\ge 1$ there would be more than $c \cdot \ln n \cdot \deg_i$ cluster centers $r_C$ of  unsampled clusters $C \in \chP_i$ at pairwise distance at most $\delta_i$.
Applying  Lemma \ref{lm:property} to any of them we conclude that that the particular cluster $C$ was superclustered  by a nearby sampled cluster, i.e., $C \nin \chU_i$, contradiction. In the case $i=0$, we would have that $v$ has at least $c \cdot \ln n \cdot \deg_i$ unsampled neighbors, which occurs with probability at most $n^{-c}$.
Hence, by union-bound, every vertex $v$ is explored by at most $c \cdot \ln n \cdot \deg_i$ explorations, with probability at least $1 - n^{-(c-1)}$.

Lemma \ref{lm:explorations} suggests that the interconnection step of phase $i$ can be carried out in expected $O(|E| \cdot \deg_i)$ time. Clearly, the superclustering step can be carried out in just $O(|E|)$ time, and thus the running time of the interconnection step dominates the running time of phase $i$. In order to control the running time, we terminate stage 1 and move on to stage 2 when $i_0=\lfloor\log(\kappa\rho)\rfloor$, so that $\deg_{i_0}\le n^\rho$.

Observe also that the superclustering step inserts into the spanner at most $O(n)$ edges (because we insert a subset of edges of $F_i$, and $F_i$ is a forest),
and by Lemma \ref{lm:explorations} the interconnection step inserts in expectation at most $O(|\chP_i| \cdot \deg_i \cdot ((1/\epsilon)^i + R_i) ) = O(|\chP_i| \cdot \deg_i \cdot (1/\epsilon)^i )$ edges.
(We will soon show that $R_i  = O((1/\epsilon)^{i-1})$.)
A more detailed argument providing an upper bound on the number of edges inserted by the interconnection step
will be given below.
\begin{lemma}
\label{lm:clusters}
For all $i$, $1 \le i \le \ell$,
for every pair of clusters $C \in \chU_i$, $C' \in \chP_i$ at distance at most $\half (1/\epsilon)^i$ from one another, a shortest path between the cluster centers of $C$ and $C'$ was inserted into the spanner $H$.
Moreover, for any pair $\{v\} \in \chU_0$, $\{v'\} \in \chP_0$, such that $e = (v,v') \in E$, the edge $e$ belongs to $H$.
\end{lemma}
\begin{proof}
We start with proving the first assertion of the lemma. For some index $i$, $1 \le i \le \ell$, and a pair $C \in \chU_i$, $C' \in \chP_i$ of clusters,
let $r_C,r_{C'}$ be the respective cluster centers. Then
we have
\begin{eqnarray*}
d_G(r_C,r_{C'}) &\le& \Rad(C) + d_G(C,C') + \Rad(C')\\
&\le& d_G(C,C') + 2\cdot R_i\\
&\le& \half (1/\epsilon)^i + 2 \cdot R_i = \half  \delta_i,
\end{eqnarray*}
 and so a shortest path between   $r_C$ and $r_{C'}$ was inserted into the spanner $H$.

The second assertion of the lemma is guaranteed by the interconnection step of phase 0.
\end{proof}

Next we analyze the radii of clusters' collections $\chP_i$, for $i = 1,2,\ldots$. 

\begin{lemma}
\label{lm:Ri}
For $i = 0.1,\ldots,\ell$, the value of $R_i$ is given by
$$R_i = \sum_{j=0}^{i-1} (1/\epsilon)^j \cdot 5^{i-1- j} .$$
\end{lemma}
\begin{proof}
The proof is by induction of the index $i$.
The basis ($i=0$) is immediate as $R_0=0$.
For the induction hypothesis, note that
\begin{eqnarray*}
R_{i+1} &=&  \delta_i +  R_i = (1/\epsilon)^i + 5 \cdot R_i\\
&=& (1/\epsilon)^i + 5\cdot \left(\sum_{j=0}^{i-1} (1/\epsilon)^j \cdot 5^{i-1 - j} \right) \\
& = & \sum_{j=0}^i (1/\epsilon)^j \cdot 5^{i-j}~,
\end{eqnarray*}

as required.
\end{proof}

Observe that Lemma \ref{lm:Ri} implies that for $\epsilon < 1/10$, we have  $R_i =5^{i-1}\cdot\frac{1/(5\epsilon)^i-1}{1/(5\epsilon)-1}\le {1 \over {1 - 5\epsilon}} \cdot(1/\epsilon)^{i-1} \le 2 \cdot (1/\epsilon)^{i-1}$.
Recall that $\chP_i = \chS_{i-1}$.
Hence $\Rad(\chP_i) = \Rad(\chS_{i-1})$.
By inequality (\ref{eq:rad}), we have $\Rad(\chP_i) \le (1/\epsilon)^{i-1} + 5 \cdot R_{i-1} = R_i$, for all $i = 0,1,\ldots,\ell$.

We analyze the number of clusters in collections $\chP_i$ in the following lemma.

\begin{lemma}
\label{lm:Pi}
For $i = 0,1,\ldots,i_0$,
\begin{equation}
\label{eq:Pi}
|\chP_i| ~\le~ 2 \cdot n^{1 - {{2^i -1} \over \kappa}} ~,
\end{equation}
with probability at least $1 - \exp\{-\Omega(n^{1 - {{2^i -1} \over \kappa}} )\}$.
\end{lemma}
\begin{proof}
The probability that a vertex $v\in V$ will be a center of a cluster in $\chP_i$ is $\prod_{j=0}^{i-1}1/\deg_j=n^{-(2^i-1)/\kappa}$. Thus the expected size of $\chP_i$ is $n^{1-(2^i-1)/\kappa}$, and by Chernoff bound,
\begin{eqnarray*}
\Pr[|\chP_i|\ge 2\E[|\chP_i|]]\le \exp\{-\Omega(\E[|\chP_i|])\}=\exp\{-\Omega(n^{1 - {{2^i -1} \over \kappa}})\}~.
\end{eqnarray*}



\end{proof}


Since for $\rho < 1/2$ and $i \le i_0 = \lfloor \log (\kappa\rho) \rfloor$, we have
$n^{1 - {{2^{i+1} - 1} \over \kappa}}  \ge n^{1-2\rho} = \omega(\log n)$, we conclude that whp for all $0\le i\le i_0$,
$|\chP_{i+1}| =  |\chS_i| = O( n^{1 - {{2^{i+1} - 1} \over \kappa}})$.
Hence in particular, $|\chP_{i_0+1}| = O( n^{1 - \rho +1/\kappa})$, whp.

The total expected  running time of the first stage is  at most
$$O(|E|) \sum_{i=1}^{i_0} \deg_i ~=~ O(|E|  \cdot n^{{2^{i_0} } \over \kappa})  ~=~  O(|E| \cdot n^\rho)~.$$

Since each superclustering step inserts at most $O(n)$ edges into the spanner,
the overall number of edges inserted by the $i_0$ superclustering steps of stage 1 is $O(n \cdot i_0) = O(n \cdot \log (\kappa\rho))$.
The expected number of edges added to the spanner by the interconnection step of phase $i$ is at most
\begin{eqnarray*}
O(|\chP_i| \cdot \deg_i  \cdot  (1/\epsilon)^i) = O(n^{1 - {{2^i - 1} \over \kappa}} \cdot  (1/\epsilon)^{\log (\kappa\rho)} \cdot n^{2^i \over \kappa})= O(n^{1+ 1/\kappa}  \cdot (1/\epsilon)^{\log (\kappa\rho)})~.
\end{eqnarray*}

Next we describe stage 2 of the algorithm, i.e., phases $i = i_0+1,i_0+2,\ldots,i_1$, where $i_1 = i_0 + \lceil {{\kappa + 1} \over {\kappa\rho}} \rceil - 2$.
All these phases are executed with the same fixed degree parameter $\deg_i = n^\rho$. On the other hand, the distance threshold keeps growing in the same steady rate as in stage 1, i.e., it is given by $\delta_i = (1/\epsilon)^i + 4 R_i$.
The sets $\chP_{i_0+1},\chP_{i_0+2},\ldots,\chP_{i_1}$ on which phases $i_0+1,i_0+2,\ldots,i_1$, respectively, operate are defined by
$\chP_{i_0+i} ~=~ \chS_{i_0+i-1}$,
for $1 \le i \le i_1 - i_0$.

Lemma \ref{lm:Ri}  keeps holding for these additional $i_1 - i_0$ phases, i.e.,
$$\Rad(\chP_i) ~\le~ R_i ~\le ~ 2 \cdot (1/\epsilon)^{i-1}~.$$
(We assume all through that $\epsilon < 1/10$.)

Also, for every pair of
clusters $C \in \chU_i$ and $C' \in \chP_i$ which are at
 distance at most $\half (1/\epsilon)^i$ from one another, their centers are interconnected in the spanner by a shortest path between them.

In addition, for every $i \in [i_0,i_1]$, the expected size of $\chP_{i+1}$ is
\[
\E[|\chP_{i+1}|]=n\cdot\prod_{j=0}^{i}1/\deg_j\le n^{1 + 1/\kappa  - (i + 1 - i_0) \rho}~.
\]
By Chernoff bound, for every such $i$, with probability at least $1 - \exp\{-\Omega(n^\rho)\}$,
we have
$$|\chP_{i+1}|= |\chS_i| ~\le~ 2 \cdot n^{1 + 1/\kappa -\rho - (i - i_0)\rho} ~.$$

Assuming that $n^\rho=\omega(1)$, we conclude that whp
\begin{eqnarray}
\label{eq:finalPi}
|\chP_{i_1 + 1}| = |\chS_{i_1}| ~\le~ O(n^{1 + 1/\kappa  - (i_1 + 1 - i_0) \rho})
= O(n^{1 +1/\kappa - (\lceil {{\kappa + 1} \over {\kappa\rho}} \rceil - 1)\rho} ) ~=~ O(n^\rho)~.
\end{eqnarray}
(For the assumption above to hold we will need to assume that $\rho \ge {{\log\log n} \over {2\log n}}$, say. We will show soon that this assumption is valid in our setting.)

The time required to perform these $\lceil {{\kappa + 1} \over {\kappa\rho}} \rceil - 2\le 1/\rho$ additional phases is expected to be at most $O(|E| \cdot \deg_i  \cdot (1/\rho)) = O(|E| n^\rho/\rho)$.

The final collection of clusters $\chP_{i_1 + 1}$ is created by setting
$\chP_{i_1 + 1} ~=~ \chS_{i_1}$.

We will next bound the expected number of edges inserted into the spanner during stage 2 of the algorithm. Each of the forests $F_i$, $i \in [i_0+1,i_1]$, created during the superclustering steps contributes at most $n-1$ edges.
The interconnection step of phase $i+i_0$ contributes in expectation
at most $O(|\chP_{i+i_0}| \cdot \deg_{i+i_0} \cdot (1/\epsilon)^i ) \le O(n^{1+1/\kappa-i\rho} \cdot (1/\epsilon)^{\log(\kappa\rho)+i})$ edges. Assuming that $1/\epsilon<n^{\rho}/2$, this becomes a geometric progression, so the overall expected number of edges inserted into the spanner on stage 2 is $O(n^{1+1/\kappa} \cdot (1/\epsilon)^{\log (\kappa\rho)})$. (We will show the validity of this assumption in the end of this section.)

Finally, we describe the concluding phase of the algorithm, i.e., phase $\ell = i_1 +1$. In this phase we skip the superclustering step (as the number of clusters is already sufficiently small), and proceed directly to the interconnection step.

On this  step each of the cluster centers  $r_C$ for $C \in \chP_\ell$ conducts a BFS exploration in $G$ to depth $\half \delta_\ell = \half (1/\epsilon)^\ell + 2R_\ell$.
(Essentially, we define $\chU_\ell = \chP_\ell$, and perform the usual interconnection step of the algorithm.)
By (\ref{eq:finalPi}), the number of edges inserted by this step into the spanner is whp only
$O( |\chP_\ell|^2 \cdot (1/\epsilon)^\ell) = O(|\chP_{i_1 +1}|^2 \cdot (1/\epsilon)^{i_1+1}) = O(n^{2\rho}  \cdot (1/\epsilon)^{\log (\rho\kappa) + 1/\rho})$.
Recall that we assume that $\rho < 1/2$. Hence this number of edges is sublinear in $n$.

Hence the overall expected number of edges in the spanner is $|H| = O(n^{1+1/\kappa}   \cdot (1/\epsilon)^{\log (\kappa\rho)})$. Observe also that the running time of the last phase is $O(|E| \cdot n^\rho )$.
Hence the overall expected running time of the algorithm is $O(|E| n^\rho/\rho)$.
It remains to analyze the stretch of the resulting spanner $H$.

Let $\chU = \bigcup_{j=0}^\ell \chU_j$.
Observe that every singleton cluster $\{v\} \in \chP_0$ is associated with exactly  one cluster of $\chU$, i.e., $\chU$ is a partition of $V$.
Note that $\Rad(\chU_0) = 0$, $\Rad(\chU_1) \le 1=R_1$, and for every $j \in [\ell]$,
 we have $\Rad(\chU_j) \le \Rad(\chP_j) \le R_j \le 2 \cdot (1/\epsilon)^{j-1}$. Denote $c= 2$.
Recall also (see Lemma \ref{lm:clusters}) that for $j \ge 1$, for  every  pair of clusters $C,C' \in \chU_j$ at distance at most $\half (1/\epsilon)^j$ from one another,  a shortest path between the cluster centers of this pair of clusters in $G$ was added to the spanner $H$.
Moreover, neighboring clusters of $\chU_0$ are also interconnected by a spanner edge.

\begin{lemma}
\label{lm:stretch_aux}
Consider  a pair of indices $0 \le j < i \le \ell$, and a pair of neighboring (in $G$) clusters $C' \in \chU_j$, $C \in \chU_i$, and a vertex $w' \in C'$ and the center $r$ of $C$.
Then  the spanner $H$ contains a path of length at most $3 \Rad(\chU_j) + 1 + \Rad(\chU_i)$ between $w'$ and $r$.
\end{lemma}
\begin{proof}
Let $(z',z) \in E \cap (C' \times C)$ be an edge connecting this pair of clusters. There exists a subcluster $C'' \subseteq C$, $C'' \in \chP_j$ such that $z \in C''$. Hence the interconnection step of phase $j$ inserted a shortest path $\pi(r',r'')$
in $G$ between the cluster centers $r'$ of $C'$ and $r''$ of $C''$ into the spanner $H$.
Note that the distance between $r',r''$ is at most $1+2\Rad(\chP_j)\le 1+4(1/\epsilon)^{j-1}<1/2\cdot(1/\epsilon)^j$, since we assume $\epsilon<1/10$.
 Hence a path between $w'$ and $r$ in $H$ can be built by concatenating a path $\pi(w',r')$ between $w'$ and $r'$ in the spanning tree $T(C')$ of $C'$ with the path $\pi(r',r'')$ in $H$, and
with the path $\pi(r'',r)$ in the spanning tree $T(C)$ of $C$. (Note that both $r''$ and $r$ belong to $C$.) Its length is at most
\begin{eqnarray*}
|\pi(w',r')| + |\pi(r',r'')| + |\pi(r'',r)|&\le& \Rad(C') + (\Rad(C') + 1 + \Rad(C'')) + \Rad(C)\\& \le&  3 \Rad(\chU_j) + 1 + \Rad(\chU_i)~.
\end{eqnarray*}
\end{proof}

Now we are ready to analyze the stretch of our spanner.

\begin{lemma}
\label{lm:stretch}
Suppose $\epsilon \le 1/10$.
Consider a pair of vertices $u,v \in V$. Fix a shortest path $\pi(u,v)$ between them in $G$, and suppose that for some index $i \in [0,\ell]$, all vertices of $\pi(u,v)$ are clustered in the set $\chU^{(i)}$ defined by $\chU^{(i)} = \bigcup_{j=0}^i\chU_j$.
Then
$$
d_H(u,v) \le (1  + 16 c \cdot \epsilon \cdot i ) d_G(u,v)  + 4 \sum_{j=1}^i R_j \cdot 2^{i-j}~.
$$
\end{lemma}
\begin{proof}
The proof is by induction on $i$.
For the induction basis $i=0$, observe that all vertices of $\pi(u,v)$ are clustered in $\chU_0$, and thus all edges of $\pi(u,v)$ are inserted  into the spanner on phase 0.
Hence $d_H(u,v) = d_G(u,v)$.

For the induction step, consider first a pair of vertices $x,y$ such that $|\pi(x,y)| \le \half (1/\epsilon)^i$,
and $V(\pi(x,y)) \subseteq \chU^{(i)}$.
 Let $z_1$ and $z_2$ be the leftmost and the rightmost $\chU_i$-clustered vertices in $\pi(x,y)$, if exist. (The case when both these vertices exist is the one where the largest stretch is incurred; cf. \cite{EP04}.)
Let $C_1,C_2 \in \chU_i$ be their respective clusters, i.e., $z_1 \in C_1$, $z_2 \in C_2$.
Let $w_1$ (respectively, $w_2$) be the neighbor of $z_1$ (resp., $z_2$) on the subpath $\pi(x,z_1)$ (resp., $\pi(z_2,y)$) of $\pi(x,y)$, and denote by $C'_1$ and $C'_2$
the respective clusters of $w_1$ and $w_2$. Observe that $C'_1,C'_2 \in \chU^{(i-1)}$.

Denote $r_1$ and $r_2$ the  cluster centers of $C_1$ and $C_2$, respectively. The spanner $H$ contains a path of length at most $d_G(r_1,r_2)$ between these cluster centers.
Also, by Lemma \ref{lm:stretch_aux}, since $C'_1$ and $C_1$ are neighboring clusters, the spanner $H$ contains a path of length at most
$3 R_j+ 1 + R_i \le 2R_i +1$
between $w_1$ and $r_1$, and a path of at most this length between $r_2$ and $w_2$.
(For $\epsilon < 1/10$, $3R_j \le R_i$, for all $j < i$.)
Observe also that the subpaths $\pi(x,w_1)$ and $\pi(w_2,y)$ of $\pi(x,y)$ have all their vertices clustered in $\chU^{(i-1)}$, and thus the induction hypothesis is applicable to these subpaths.

Hence
\begin{eqnarray*}
d_H(x,y)
 &\le&  d_H(x,w_1) + d_H(w_1,r_1) + d_H(r_1,r_2) + d_H(r_2,w_2) + d_H(w_2,y) \\
& \le & (1  + 16 c \cdot \epsilon (i-1)) d_G(x,w_1) +  4 \sum_{j=1}^{i-1} R_j \cdot 2^{i-1 - j} + 2R_i + 1 + (d_G(C_1,C_2) + 2R_i) \\&& + 2R_i + 1  +   (1 + 16c \cdot \epsilon(i-1)) \cdot d_G(w_2,y) + 4 \sum_{j=1}^{i-1} R_j \cdot 2^{i-1-j} \\
& = &  (1 + 16c \cdot \epsilon(i-1)) \cdot (d_G(x,w_1) + d_G(w_2,y)) + d_G(C_1,C_2) +   4R_i + 2 + 8\sum_{j=1}^{i-1}R_j \cdot 2^{i-1-j}.
\end{eqnarray*}
Note also that
\begin{eqnarray*}
d_G(x,y)& = & d_G(x,w_1) + 1 + d_G(z_1,z_2) +1 + d_G(w_2,y) \ge  d_G(x,w_1)  + d_G(C_1,C_2) + 2 + d_G(w_2,y).
\end{eqnarray*}
Hence
\begin{eqnarray}\nonumber
d_H(x,y)\nonumber &\le & (1 + 16c \cdot \epsilon(i-1)) d_G(x,y) + 4R_i  +8\sum_{j=1}^{i-1}R_j \cdot 2^{i-1-j} \\\nonumber
\label{eq:stretch}
& = & (1  + 16c \cdot \epsilon(i-1)) d_G(x,y) + 4 \sum_{j=1}^i R_j \cdot 2^{i-j}.
\end{eqnarray}
Now consider a pair of vertices $u,v$ such that all vertices of $\pi(u,v)$ are clustered in $\chU^{(i)}$, without any restriction on $|\pi(u,v)|$.
We partition $\pi(u,v)$ into segments $\pi(x,y)$ of length exactly $\lfloor \half (1/\epsilon)^i \rfloor$,  except maybe one segment of possibly smaller length. Inequality (\ref{eq:stretch}) applies to all these segments.
Hence
\begin{eqnarray*}
d_H(u,v) &\le & (1  + 16c \cdot \epsilon(i-1))d_G(u,v) + 4 \sum_{j=1}^i R_j \cdot 2^{i-j} \lfloor {{d_G(u,v)} \over {\half (1/\epsilon)^i - 1}} \rfloor + 4\sum_{j=1}^i R_j \cdot 2^{i-j} \\
& \le &  \left(1  + 16c \cdot \epsilon(i-1) + {{4\sum_{j=1}^{i} R_j \cdot 2^{i-j}} \over {\half (1/\epsilon)^i - 1}}\right) d_G(u,v) + 4\sum_{j=1}^i R_j \cdot 2^{i-j}.
\end{eqnarray*}
It remains to argue that ${{8 \sum_{j=1}^i R_j 2^{i-j}} \over {(1/\epsilon)^i - 2}} \le 16c\cdot \epsilon$.
Recall that for every $j$, we have $R_j \le c \cdot (1/\epsilon)^{j-1}$. Since $1/\epsilon \ge 10$,
the left-hand-side is at most
\begin{eqnarray*}
10 c  \cdot  \epsilon^i \sum_{j=1}^i  (1/\epsilon)^{j-1} \cdot 2^{i-j}= 10c \cdot \epsilon \sum_{j=1}^i (2\epsilon)^{i-j}=10 c\cdot \epsilon \sum_{h=0}^{i-1} (2\epsilon)^h ~\le 16c\cdot \epsilon~.
\end{eqnarray*}

\end{proof}

Observe that (as $\epsilon \le 1/10$), we have
\begin{eqnarray*}
\sum_{j=1}^i R_j \cdot 2^{i-j} &\le&  c \sum_{j=1}^i(1/\epsilon)^{j-1} \cdot 2^{i-j} ~=~ c \cdot 2^{i-1} \cdot \sum_{j=1}^i \left({1 \over {2\epsilon}}\right)^{j-1} \\
& = &
c \cdot 2^{i-1} \sum_{j=0}^{i-1} \left({1 \over {2\epsilon}} \right)^{j} ~=~ c \cdot 2^{i-1} {{ (1/2\epsilon)^{i} - 1} \over {(1/2\epsilon) - 1}} \\
& = & c \cdot \epsilon \cdot {{\half \cdot (1/\epsilon)^i - 2^{i-1}} \over { \half - \epsilon}} ~=~
O\left(\left({1 \over \epsilon}\right)^{i-1}\right)~.
\end{eqnarray*}
Note also that the condition of the last lemma holds with $i = \ell$ for every pair $u,v \in V$ of vertices. Hence

\begin{corollary}
\label{cor:stretch}
For every pair $u,v \in V$,
$$
d_H(u,v) ~\le~ (1+  16 c \cdot \ell\cdot \epsilon)d_G(u,v) + O((1 / \epsilon)^{\ell-1})~.
$$
\end{corollary}

Recall that the spanner $H$ contains, whp, $|H| = O(n^{1+1/\kappa} \cdot \log n \cdot (1/\epsilon)^{\log (\kappa\rho)} + n \cdot \log n \cdot (1/\epsilon)^{\log (\kappa\rho) + 1/\rho})$ edges, and the expected running time required to construct it is $O(|E| \cdot n^\rho /\rho )$.
Recall also that $\ell = i_1 + 1 \le \log (\kappa\rho)  + 1/\rho +1$.
Set now $\epsilon' = 16c \cdot \ell \cdot \epsilon$. We obtain stretch $\left(1+\epsilon',O\left({{\log \kappa + 1/\rho} \over {\epsilon'}}\right)^{\log \kappa + 1/\rho}\right)$. The condition $\epsilon < 1/10$ translates now to $\epsilon' \le 1.6 c(\log (\kappa \rho)+ 1/\rho)$. We will replace it by a simpler stronger condition $\epsilon \le 1$.

\begin{corollary}
\label{cor:sp}
For any parameters $0 < \epsilon \le 1$,  $\kappa\ge 2$, and $\rho >0$,  and any $n$-vertex unweighted graph $G = (V,E)$, our algorithm computes a $(1+\epsilon,\beta)$-spanner with expected number of edges
$O\left({{\log \kappa + 1/\rho} \over \epsilon}\right)^{\log \kappa}\cdot n^{1+1/\kappa}   $, in expected time $O(|E|\cdot n^\rho/\rho)$, where
$$\beta = \left({{O(\log \kappa + 1/\rho)} \over \epsilon}\right)^{\log \kappa + 1/\rho}~.$$
\end{corollary}

A particularly useful setting of parameters is $\rho = 1/\log \kappa$.
Then we get a spanner with expected $ O\left({{ \log \kappa} \over {\epsilon}}\right)^{\log\kappa}  \cdot n^{1+1/\kappa} $ edges,
in time $O(|E| \cdot n^{1 \over {\log\kappa}} \cdot \log \kappa )$, and
$\beta = O\left({{\log \kappa} \over \epsilon}\right)^{2\log \kappa}$.

We remark that it makes no sense to set $\rho < 1/\kappa$, as the resulting parameters will be strictly worse than when $\rho =1/\kappa$. Also our assumptions that $\rho>\log\log n/(2\log n)$ and $16c \cdot \ell /\epsilon<n^\rho/2$ are justified, as otherwise we get $\beta\ge n$, so a trivial spanner will do.

\subsection{An Improved Variant of the Algorithm}
\label{sec:sp_impr}

In this section we show that the leading coefficient $O((\log \kappa + 1/\rho)/\epsilon)^{\log \kappa}$ of $n^{1+1/\kappa}$ in the size of the spanner can be almost completely eliminated at essentially no price. We also devise  here  yet sparser constructions of emulators.



For $i = 0,1,\ldots,\ell$, denote by $N_i = |\chP_i|$ the expected number of clusters which take part in phase $i$.
Recall also that the interconnection step of the $i$th phase contributes $N_i \cdot \deg_i  \cdot \left({{c'\ell}  \over \epsilon}\right)^i$ edges in expectation, where $\ell$ is the total number of steps, and $c'$ is a universal constant. Note that the contribution of the interconnection step dominates the contribution of the superclustering step in the current variant of the algorithm, and it will still be the case after the modification that we will now introduce.
Hence we will now focus on decreasing the number of edges contributed by the interconnection steps.

We keep the structure of the algorithm intact, and have the values of distance thresholds $\delta_i$ unchanged. The only change is in the degree sequence $\deg_0,\deg_1,\ldots$ of degree parameters used in phases $0,1,\ldots$, respectively.
Next, we describe our new setting of these parameters for stage 1 of the algorithm (i.e., phases $i$, $1 \le i \le i_0$). In the case that $\kappa\ge 16$ let $a=\log\log\kappa$, otherwise, when $\kappa<16$, let $a=2$. Define $i_0=\min\{\lfloor\log (a\kappa \rho) \rfloor,\lfloor\kappa \rho \rfloor\}$, and for $i=0,1,\ldots,i_0$ let $\deg_i=n^{(2^i-1)/(a\kappa)+1/\kappa}$. We now have that for $i\le i_0+1$,
\[
N_i=n\prod_{j=0}^{i-1}1/\deg_j=n^{1-\frac{2^i-1-i}{a\kappa}-\frac{i}{\kappa}},
\]
and in particular, when $i_0=\lfloor\log (a\kappa \rho) \rfloor$ we have $N_{i_0+1}\le n^{1-\frac{a\kappa\rho-1-(i_0+1)}{a\kappa}-\frac{i_0+1}{\kappa}}\le n^{1-\rho}$ (since $a\ge 2$ and $i_0 \ge 1$). Whenever $i_0=\lfloor\kappa \rho \rfloor$ we also have $N_{i_0+1}\le n^{1-\frac{i_0+1}{\kappa}}\le n^{1-\rho}$. Additionally,
we always have
\[
N_i\cdot\deg_i = n^{1 - {{2^i - 1 - i} \over {a \kappa}} - { i \over \kappa}} \cdot n^{{{2^i - 1} \over {a \kappa}} + {1 \over \kappa}} =  n^{1+\frac{i}{a\kappa}-\frac{i-1}{\kappa}}~.
\]

We restrict ourselves to the case that
\begin{equation}\label{eq:rree}
\frac{c'\ell}{\epsilon}\le n^{\frac{1}{2\kappa}}/2~,
\end{equation}
which holds whenever $\kappa\le \frac{c_0\cdot\log n}{\log(\ell/\epsilon)}$, for a sufficiently small constant $c_0$. Now the expected number of edges inserted at phase $i\le i_0$ is at most
\begin{equation}\label{eq:rree1}
N_i\cdot\deg_i \cdot\left({{c'\ell}  \over \epsilon}\right)^i\le
n^{1 + {i \over {2\kappa}} - {{i-1} \over \kappa}} \cdot \left({{n^{1/(2\kappa)}} \over 2} \right)^i
= n^{1+\frac{1}{\kappa}}/2^i~.
\end{equation}
Thus the total expected number of edges inserted in the first stage is $O(n^{1+1/\kappa})$. The second stage proceeds by setting $\deg_{i_0+1}=n^{\rho/2}$, and in all subsequent phases $i_0 + i$, with $i=2,3,\ldots,i_1- i_0$, we have $\deg_i=n^\rho$ as before.
The "price" for reducing the degree in the first phase of stage two is that the number of phases $i_1$ may increase by an additive 1. It follows that $N_{i_0+1}\cdot\deg_{i_0+1}\le n^{1-\rho/2}$. For $i\ge 2$, at phase $i_0+i$ we have $N_{i_0+i}\le n^{1-3\rho/2-(i-2)\rho}$. We set $i_1=\lfloor 1/\rho\rfloor$, so that $N_{i_1+1}\le n^{\rho/2}$, and whp we have that $N_{i_1+1}\le 2n^{\rho/2}$. We calculate
\[
N_{i_0+i}\cdot\deg_{i_0+i}\le n^{1-\rho/2-(i-2)\rho}~.
\]

Note that $i_0/(2\kappa)\le\rho/2$, which holds since $i_0\le\lfloor\kappa \rho \rfloor$. Hence the condition \eqref{eq:rree} implies that $(c'\ell/\epsilon)^{i_0}\le n^{\frac{i_0}{2\kappa}}\le n^{\rho/2}$. The total expected number of edges  inserted at phase $i_0+1$ is at most
\[
N_{i_0+1}\cdot\deg_{i_0+1}\cdot \left({{c'\ell}  \over \epsilon}\right)^{i_0+1}\le n^{1-\rho/2}\cdot n^{\rho/2}\cdot{{c'\ell}  \over \epsilon}\le n^{1+1/\kappa}~.
\]
The expected contribution of phase $i_0+i$ for $i\ge 2$ is at most
\begin{eqnarray}\label{eq:rree2}
N_{i_0+i}\cdot\deg_{i_0+i}\cdot \left({{c'\ell}  \over \epsilon}\right)^{i_0+i}\le n^{1-\rho/2-(i-2)\rho}\cdot n^{\rho/2}\cdot n^{i/(2\kappa)}/2^i\le n^{1+1/\kappa}/2^i~,
\end{eqnarray}
where the last inequality uses that $\rho\ge 1/\kappa$ (which we may assume w.l.o.g). This implies that the expected number of edges in all these $\lfloor 1/\rho\rfloor$ phases is $O(n^{1+1/\kappa})$.

The upper bound on $\kappa$ under which this analysis was carried out is $\frac{c_0\cdot\log n}{\log(\ell/\epsilon)}\ge \frac{\Omega(\log n)}{\log(1/\epsilon)+\log(1/\rho)+\log^{(3)}n}$.\footnote{We denote $\log^{(k)}n$ as the iterated logarithm function, e.g. $\log^{(3)}n=\log\log\log n$.}  We summarize this discussion with the following theorem.

\begin{theorem}
\label{thm:sp}
For any unweighted graph $G = (V,E)$ with $n$ vertices , $ 0 < \epsilon < 1/10$, $2\le \kappa  \le\frac{c\cdot\log n}{\log(1/\epsilon)+\log(1/\rho)+\log^{(3)}n}$ for a constant $c$, and $1/\kappa\le\rho<1/2$, our algorithm computes a $(1 + \epsilon,\beta)$-spanner with
$\beta = O({{1} \over {\epsilon}} (\log \kappa  +1/\rho))^{\log \kappa + 1/\rho+\max\{1,\log^{(3)} \kappa\}}$
and expected number of edges $O(n^{1+1/\kappa})$. The expected running time is  $O(|E| \cdot n^\rho/\rho)$.
\end{theorem}

Note that the sparsest this spanner can be is $O(n \log\log n)$, and at this level of sparsity its $\beta = O(\log\log n+1/\rho)^{\log\log n+1/\rho}$.
(To get this bound we set $\epsilon > 0$ to be an arbitrary small constant, and $\kappa = {{c_0\log n} \over {\log^{(3)} n}}$.)

This is sparser than the state-of-the-art efficiently-computable sparsest $(1+\epsilon,\beta)$-spanner
due to \cite{P10}, which has $O(n (\log\log n)^\phi)$ edges, where $\phi = {{1+ \sqrt{5} } \over 2}$ is the golden ratio.
Moreover, this spanner has a smaller $\beta$ than the one of \cite{P10} in its sparsest level. Denoting the latter as $\beta_{Pet}$, it holds that
  $\beta_{Pet} \approx O(\log \kappa)^{1.44 \log \kappa  +1/\rho}$, i.e.,  for every setting of the  time parameter $\rho$, the exponent of our $\beta$ is smaller than that of  $\beta_{Pet}$.

\subsubsection{Sparse Emulator}\label{sec:sparse-emul}

Finally, we note that if one allows an {\em emulator} instead of spanner, then we can decrease the size all the way to $O(n)$ when $\kappa = \log n$.
To achieve this, we insert single "virtual" edges instead of every path (of length $(c'\ell/\epsilon)^i$) between every pair of cluster centers that we choose to interconnect on phase $i$, for every $i$. Analogously, in the superclustering step we also form a supercluster around a center $r_C$ of a cluster $C$ by adding virtual edges $(r_C,r_{C'})$ for each  cluster $C'$ associated with $C$. The weight of each such edge is defined by $\omega(r_C,r_{C'}) = d_G(r_C,r_{C'})$. The condition  \eqref{eq:rree} was required to obtain converging sequences at \eqref{eq:rree1} and
\eqref{eq:rree2}, but without the $(c'\ell/\epsilon)^i$ terms, the number of edges already forms a converging sequence at each stage.

Moreover, one can also use for emulators a shorter degree sequence than the one we used for spanners, and as a result to save the additive term of $\log^{(3)} \kappa$ in the exponent of $\beta$. Specifically, one can  set $\deg_i = n^{{2^i} \over \kappa}/2^{2^i-1}$, for each $i = 0,1,\ldots,i_0=\lfloor\log(\kappa\rho)\rfloor$.
As a result we get $N_i=n\cdot\prod_{j=0}^{i-1}1/\deg_j=n^{1-\frac{2^i-1}{\kappa}}\cdot 2^{2^i-1-i}$, and thus the expected number of edges inserted at phase $i\le i_0$ is at most
\[
N_i\cdot\deg_i=n^{1+1/\kappa}/2^i~.
\]
As before, when the first stage concludes, we run one phase with $\deg_{i_0+1} = n^{\rho/2}$, and all subsequent phases with $\deg_i = n^\rho$. To bound the expected number of edges added at phase $i_0+1$ we need to note that $2^{2^{i_0+1}}\le 2^{2\kappa\rho}\le n^{\rho/2}$ as long as $\kappa\le (\log n)/4$. (The latter can be  assumed without affecting any of the parameters by more than a constant factor). It follows that $N_{i_0+1}\cdot\deg_{i_0+1}= n^{1-\frac{2^{i_0+1}-1}{\kappa}}\cdot 2^{2^{i_0+1}-1-(i_0+1)}\cdot n^{\rho/2}\le n^{1+1/\kappa}$. In the remaining phases $N_{i_0+i}\le n^{1+1/\kappa-(i-1)\rho}$ for $i\ge 2$, and the contribution of these phases is a converging sequence. We conclude the discussion with the following theorem.

\begin{theorem}
\label{thm:emul}
For any unweighted graph $G = (V,E)$ with $n$ vertices, and for any parameters $0 < \epsilon \le 1$, $2\le \kappa \le (\log n)/4$, $1/\kappa\le\rho<1/2$, our algorithm computes a $(1 + \epsilon,\beta)$-emulator with
$\beta = O\left({{\log \kappa + 1/\rho} \over {\epsilon}}\right)^{\log \kappa  + 1/\rho} $
and expected number of edges $O(n^{1+1/\kappa})$. The expected running time is $O(|E| \cdot n^\rho/\rho )$.
\end{theorem}
In particular, the algorithm produces a {\em linear-size} $(1+\epsilon,\beta)$-emulator with $\beta = O\left({{\log\log n+1/\rho} \over \epsilon}\right)^{\log\log n + 1/\rho}$ within this running time.

\subsection{Distributed and Streaming Implementations}
\label{sec:distr_sp}

In this section we provide  efficient distributed and streaming algorithms for constructing sparse $(1+\epsilon,\beta)$-spanners. The distributed algorithm works in the CONGEST model.

To implement phase 0, each vertex selects  itself into $S_0$ with probability $n^{-1/\kappa}$, i.a.r.. In distributed model vertices of $S_0$ send messages to their neighbors.
Each vertex $u$ that receives at least one message, picks an origin $v \in S_0$ of one of these messages, and joins the cluster centered at $v$. It also sends negative acknowledgements to all its other neighbors from $S_0$. All unclustered vertices $z$ insert all edges incident on them into the spanner.

It is also straightforward to implement this in $O(1)$ passes in the streaming model.

Each consecutive  phase is now also implemented in a straightforward manner, i.e., BFS explorations to depth $\delta_i$ in the superclustering steps are implemented via broadcasts and convergecasts in the distributed model, and by $\delta_i$ passes in the streaming model. In the interconnection steps, however, we need to implement many BFS explorations which may explore the same vertices.
However, by Lemma \ref{lm:property} every vertex is explored on phase $i$ by up to $O(deg_i \cdot \log n)$ explorations whp, it follows that in distributed setting this step requires, whp,  $O(\deg_i \cdot \log n \cdot \delta_i)$ time. Also note that $\sum_i\delta_i=O(\beta)$.

In the streaming model we have two possible tradeoffs. The first uses expected  $O(n \cdot \deg_i)$ space to maintain for each vertex $v$ the BFS parents and distance estimates, and requires just $\delta_i$ passes. To see that such space suffices, recall that the expected number of explorations which visit any vertex $v$ is at most $\deg_i$, by Lemma \ref{lm:explorations}. Whp, the space is $O(n \cdot \deg_i \cdot \log n) = O(n^{1+\rho} \cdot \log n)$.

The second option in the streaming algorithm is to divide the interconnection step of phase $i$ to $c \cdot \deg_i \cdot \log n$ subphases, for a sufficiently large constant $c$.
On each subphase each exploration source, which was not sampled on previous subphases, samples itself i.a.r. with probability $1/\deg_i$. Then the sampled exploration sources conduct BFS explorations to depth $\delta_i/2$.
For every vertex $v$, the expected number of explorations that traverse it on each subphase is $O(\log n)$. Moreover, by Chernoff's inequality, whp, no vertex $v$ is ever traversed by more than $O(\log n)$ explorations. (Here we take a union-bound on all vertices, all phases, and all subphases. The bad events are that some vertex is traversed by more than twice its expectation explorations on some subphase.)
Hence each subphase requires $O(\delta_i)$ passes, and whp, the space requirement is $O(n \log n)$, plus the size of the spanner.
After $c \cdot \deg_i \cdot  \log n$ subphases, whp, each exploration source is sampled on at least one of the subphases, and so the algorithm performs all the required explorations.

Finally, the stretch analysis of distributed and streaming variants of our algorithm remains the same as in the centralized case.

Hence we obtain the following distributed and streaming analogues of Theorem \ref{thm:sp}.

\begin{theorem}
\label{thm:distr_sp}
For any unweighted graph $G = (V,E)$ with $n$ vertices , $ 0 < \epsilon < 1/10$, $2\le \kappa  \le\frac{c\cdot\log n}{\log(1/\epsilon)+\log(1/\rho)+\log^{(3)}n}$ for a constant $c$, and $1/\kappa\le\rho<1/2$, our distributed  algorithm (CONGEST model) computes a $(1 + \epsilon,\beta)$-spanner with
$\beta = O({{1} \over {\epsilon}} (\log \kappa  +1/\rho))^{\log \kappa + 1/\rho+\max\{1,\log^{(3)} \kappa\}}$
and expected number of edges $O(n^{1+1/\kappa})$. The required number of rounds is whp $O(n^\rho/\rho\cdot \beta\cdot \log n)$.

Our streaming algorithm computes a spanner with the above properties, in either: $O(n \log n + n^{1+1/\kappa})$ expected space and $O(n^\rho/\rho\cdot\log n \cdot \beta)$ passes, or using $O(n^{1+\rho}\cdot \log n)$ space, whp, and $O(\beta)$ passes.
\end{theorem}

The streaming algorithm described above can also be modified to provide a $(1+\epsilon,\beta)$-emulator as in Theorem \ref{thm:emul}, within the same pass and space complexities.

\section{Applications}
\label{sec:appls}

In this section we describe applications of our improved constructions of spanners and emulators to computing approximate shortest paths for a set $S \times V$ of vertex pairs, for a subset $S \subseteq V$ of designated sources.

\subsection{Centralized Setting}
\label{sec:appl_centr}

We start with the centralized setting. Here our input graph $G = (V,E)$ is unweighted, and we construct a $(1+\epsilon,\beta)$-emulator $H$ of $G$ with  $O(n^{1+1/\kappa})$ edges, in expected time
$O(|E| \cdot n^\rho/\rho)$, where
\begin{equation}
\label{eq:beta_centr}
\beta = O\left({{\log \kappa + 1/\rho} \over \epsilon}\right)^{\log \kappa + 1/\rho}~.
\end{equation}

Observe that all edge weights in $H$ are integers in the range $[1,\beta]$. We round all edge weights up to the closest power of $1+ \epsilon$. Let $H'$ be the resulting emulator.
Note that for any pair $u,v$ of vertices, we have
\begin{eqnarray*}
d_G(u,v) &\le& d_H(u,v) ~\le~ d_{H'}(u,v) ~\le (1+\epsilon)d_H(u,v)\\
&\le& (1+\epsilon)^2 d_G(u,v) + (1+\epsilon)\beta~.
\end{eqnarray*}
For a sufficiently small $\epsilon$, $(1+\epsilon)^2 \le 3\epsilon$, and we rescale $\epsilon' = 3\epsilon$. As a result the constant factor hidden by the $O$-notation in the basis of $\beta$'s exponent grows, but other than that $H'$ has all the properties of the emulator $H$. Also, it employs only
\begin{eqnarray*}
t = O\left({{\log \beta} \over \epsilon}\right) = O\left({{(\log 1/\epsilon + \log(\log \kappa + 1/\rho))\cdot (\log \kappa + 1/\rho)} \over \epsilon}\right)
\end{eqnarray*}
 different edge weights.
Hence a single-source shortest path computation in $H$ can be performed in $O(|H| + n \log t) = O(n^{1+1/\kappa} + n(\log (1/\epsilon) + \log (\log \kappa + 1/\rho)))$ time \cite{OMSW10}. (See also \cite{KMP11}, Section 5.)
Hence computing $S \times V$ $(1+\epsilon,\beta)$-approximate shortest distances requires $O(|E| \cdot n^\rho/\rho + |S|(n^{1+1/\kappa} + n(\log (1/\epsilon) + \log (\log \kappa + 1/\rho))))$ time.

\begin{theorem}
\label{thm:asd_unwtd}
For any $n$, and for any parameters $0 < \epsilon \le 1$, $2\le\kappa\le(\log n)/4$,
$1/\kappa\le\rho \le  1/2$, and any $n$-vertex unweighted graph $G = (V,E)$ with a set $S\subseteq V$, our algorithm computes
  $(1+\epsilon,\beta)$-approximate  $S \times V$ shortest distances in the centralized model in expected
$O(|E| \cdot n^\rho/\rho + |S|(n^{1+1/\kappa} + n(\log (1/\epsilon) + \log (\log \kappa + 1/\rho))))$ time,
where $\beta$ is given by (\ref{eq:beta_centr}).
\end{theorem}

If one is interested in actual paths rather than just in distances, then one can use our $(1+\epsilon,\beta)$-spanner with
\begin{equation}
\label{eq:beta_cent_asp}
\beta = O\left({{\log \kappa + 1/\rho} \over \epsilon}\right)^{\log \kappa + 1/\rho+\max\{1,\log^{(3)} \kappa\}}~,
\end{equation}
and  $O(n^{1+1/\kappa})$ edges, but restricting $\kappa\le\frac{O(\log n)}{\log(1/\epsilon)+\log(1/\rho)+\log^{(3)}n}$.
After computing the spanner $H$ with these properties, we conduct BFS explorations on $H$ originated at each vertex of $S$.
The overall running time becomes $O(|E| \cdot n^\rho /\rho+ |S|\cdot n^{1+1/\kappa})$.

\begin{corollary}
\label{cor:asp_unwtd}
For any $n$, and for any parameters $0 < \epsilon \le 1$, $\kappa\le\frac{O(\log n)}{\log(1/\epsilon)+\log(1/\rho)+\log^{(3)}n}$,
$1/\kappa\le\rho \le 1/2$,  and any $n$-vertex unweighted graph $G = (V,E)$, our algorithm computes
$(1+\epsilon,\beta)$-approximate $S \times V$ shortest {\em paths} in the centralized model in expected
 $O(|E| \cdot n^\rho /\rho+ |S|\cdot n^{1+1/\kappa})$ time, with $\beta$ given by (\ref{eq:beta_cent_asp}).
\end{corollary}

A useful setting of parameters is $\rho = {1 \over {\log \kappa}}$. Then the running time of our algorithms from Theorem \ref{thm:asd_unwtd} and Corollary \ref{cor:asp_unwtd}  become respectively $O(|E|\cdot n^{1/\log \kappa} \cdot \log \kappa + |S|(n^{1+1/\kappa} + n(\log 1/\epsilon + \log\log \kappa)))$
and $O(|E| \cdot n^{1/\log \kappa} \cdot \log \kappa + |S| \cdot n^{1+1/\kappa} )$ (we note that the former has smaller $\beta$ given by \eqref{eq:beta_centr}, while the latter has slightly larger $\beta$ given by \eqref{eq:beta_cent_asp}).

The algorithm of Corollary \ref{cor:asp_unwtd} always outperforms the algorithm which can be derived by using the spanner of \cite{P10}  within the same scheme.
Specifically, the running time of that algorithm is at least $\tO(|E| n^\rho) + O(|S| \cdot (n^{1+1/\kappa} + n ({{\log \kappa} \over \epsilon})^\phi))$,
and the additive error $\beta_{Pet}$ there is given by
$$\beta_{Pet} = O\left({{\log \kappa  + 1/\rho} \over \epsilon}\right)^{\log_\phi  \kappa + 1/\rho}~,$$
where $\phi = {{1+\sqrt{5}} \over 2}$ is the golden ratio.
So $\beta_{Pet}$ is typically polynomially larger than the additive error $\beta$ in our algorithm, e.g., for $\rho = 1/\log \kappa$ we have $\beta_{Pet} \approx \beta^{1.22}$.
(Setting $\rho$ to be smaller than $1/\log \kappa$ makes less sense, because then the additive errors $\beta$ and $\beta_{Pet}$ deteriorate. At any rate, as $\rho$ tends to 0, the two estimates approach each other.)
Also, in the bound of \cite{P10} there is a term of $|S|\cdot n \cdot ({{\log \kappa} \over \epsilon})^\phi))$, 
which does not occur in our construction.

\subsection{Streaming Setting}
\label{sec:str_unw}

In this section we show how efficient constructions of spanners and emulators for an unweighted graph $G$  in the streaming setting can be used for efficient computation of approximate shortest paths.

First, one can use $O(n^\rho /\rho\cdot \log n \cdot \beta)$ passes  and  expected space $O(n^{1+1/\kappa} + n\cdot \log n)$ to construct an $(1+\epsilon,\beta)$-emulator $H$ with $\beta = (O(\log \kappa+ 1/\rho)/\epsilon)^{\log \kappa +1/\rho}$. One can now compute $V \times V$ $(1+\epsilon,\beta)$-approximate shortest distances in $G$ by computing exact $V\times V$ shortest distances in $H$, using the same space and without additional passes. (Observe that we do not store the output, as its size is larger than the size of $H$.) In particular, one can use here space of $O(n \cdot \log n)$, and have $\beta = (O(\log\log n + 1/\rho)/\epsilon)^{\log\log n + 1/\rho}$.
It is  also possible to  set here $\rho = {1 \over {\log\log n}}$, and obtain $2^{O({{\log n} \over {\log\log n}})}$ passes,
$\beta = O({{\log\log n} \over \epsilon})^{2\log\log n}$.

Another option is to use space $O(n^{1+\rho} \cdot \log n)$, whp, and $O(\beta)$ passes for constructing the same emulator $H$.
Again given the emulator we can compute $V \times V$ approximate shortest distances in $G$ by computing shortest distances in $H$ offline.

\begin{corollary}
\label{cor:str_unw}
For any $n$ and any parameters $0 < \epsilon \le 1$, $2\le\kappa\le(\log n)/4$, $1/\kappa\le\rho\le 1/2$, and any $n$-vertex unweighted graph $G$, our streaming algorithm computes $(1+\epsilon,\beta)$-approximate shortest {\em distances} for $V\times V$ with $\beta$ given by (\ref{eq:beta_centr}). It uses in expectation either $O(n^\rho/\rho\cdot \log n \cdot \beta)$ passes and expected space $O( n^{1+1/\kappa} + n \cdot \log n)$ or $O(\beta)$ passes
and  space $O(n^{1+\rho} \cdot \log n)$, whp.
\end{corollary}


If the actual paths rather than just distances are needed, then  we compute a $(1+\epsilon,\beta)$-spanner $H$ with $\beta$ given by (\ref{eq:beta_cent_asp}) and with expected $O(n^{1+1/\kappa})$ edges (with the restriction on $\kappa$ as above). Then we compute $V \times V$ $(1+\epsilon,\beta)$-approximate shortest paths in $H$ offline, using space $O(H|)$.

\begin{corollary}
\label{cor:str_unw_asp}
For any $n, \kappa,\epsilon, \rho$ and $G$ as in Corollary \ref{cor:asp_unwtd}, a variant of our streaming algorithm computes $(1+\epsilon,\beta)$-approximate shortest {\em paths}
for $V \times V$ with $\beta$ given by (\ref{eq:beta_cent_asp}).
It uses  either $O(n^\rho/\rho \cdot \log n \cdot \beta)$ passes  and  expected space $O(n^{1+1/\kappa}  + n \cdot \log n)$ or $O(\beta)$ passes and space $O(n^{1+\rho} \cdot \log n)$, whp.
\end{corollary}

If one is interested only in $S \times V$ paths or distances, then it is possible to eliminate the additive term of $\beta$ by using $O(|S| \cdot \beta/\epsilon)$ additional passes. These passes are used to compute exactly distances between pairs $(s,v) \in S \times V$, with $d_G(s,v) \le \beta/\epsilon$.
The overall number of passes becomes $O(|S| \cdot \beta/\epsilon + n^\rho /\rho \cdot \log n \cdot \beta)$, and space   $O( n^{1+1/\kappa} + n \cdot \log n)$.
Whenever $|S|\ge n^{1/\kappa}\log n$, we can set $\rho = {{\log |S|} \over {\log n}} - {{ \log\log n} \over {\log n}}$, and obtain the following corollary.

\begin{corollary}
\label{cor:str_unw_eps}
For any $n, \epsilon, \rho, \kappa$ and $G$ as in Corollary \ref{cor:str_unw}, and any set $S\subseteq V$ of size at least $|S|\ge n^{1/\kappa}\log n$, a variant of our streaming algorithm computes
 $(1+\epsilon)$-approximate shortest distances for $ S \times V$ in expected
\begin{equation}
\label{eq:passes}
\frac{|S|}{\epsilon} \cdot O\left(\frac{1}{\epsilon}\left(\log\kappa  + {{\log n} \over {\log |S|\! -\! \log\log n}}\right)\right)^{\log \kappa \! +\! {{\log n} \over {\log |S| \!-\! \log\log n}}}
\end{equation}
 passes, and  space $O( n^{1+1/\kappa} + n \cdot \log n)$. To compute actual paths we have similar complexities\footnote{Though there is an additional additive term of $\log^{(3)} \kappa$ in the exponent in (\ref{eq:passes}).}, but one needs to restrict $\kappa$ as in Corollary \ref{cor:asp_unwtd}.
\end{corollary}

Observe that for a constant $\epsilon>0$ and $|S| = n^{\Omega(1)}$, we can take a constant $\kappa$, so the number of passes is $O(|S|)$. One can also get for $|S| = 2^{\Omega(\log n/\log\log n)}$,
a streaming algorithm for computing $(1+\epsilon)$-approximate shortest paths for $S \times V$ by setting $\kappa = c\log n/\log\log\log n$, and obtaining
$O(|S|/\epsilon) \cdot O({{\log\log n} \over \epsilon})^{2\log\log n}$ passes and space $O(n \log n)$.

\bibliographystyle{alpha}
\bibliography{spanner}

\newcommand{\etalchar}[1]{$^{#1}$}
\begin{thebibliography}{OMSW10}

\bibitem[AB16]{AB16}
Amir Abboud and Greg Bodwin.
\newblock The 4/3 additive spanner exponent is tight.
\newblock In {\em Proceedings of the 48th Annual {ACM} {SIGACT} Symposium on
  Theory of Computing, {STOC} 2016, Cambridge, MA, USA, June 18-21, 2016},
  pages 351--361, 2016.

\bibitem[ABCP93]{ABCP93}
Baruch Awerbuch, Bonnie Berger, Lenore Cowen, and David Peleg.
\newblock Near-linear cost sequential and distribured constructions of sparse
  neighborhood covers.
\newblock In {\em 34th Annual Symposium on Foundations of Computer Science,
  Palo Alto, California, USA, 3-5 November 1993}, pages 638--647, 1993.

\bibitem[ABN11]{ABN06}
Ittai Abraham, Yair Bartal, and Ofer Neiman.
\newblock Advances in metric embedding theory.
\newblock {\em Advances in Mathematics}, 228(6):3026 -- 3126, 2011.

\bibitem[ABP17]{ABP17}
Amir Abboud, Greg Bodwin, and Seth Pettie.
\newblock A hierarchy of lower bounds for sublinear additive spanners.
\newblock In {\em Proceedings of the Twenty-Eighth Annual ACM-SIAM Symposium on
  Discrete Algorithms}, pages 568--576, 2017.

\bibitem[ACIM99]{ACIM99}
Donald Aingworth, Chandra Chekuri, Piotr Indyk, and Rajeev Motwani.
\newblock Fast estimation of diameter and shortest paths (without matrix
  multiplication).
\newblock {\em {SIAM} J. Comput.}, 28(4):1167--1181, 1999.

\bibitem[ADD{\etalchar{+}}93]{ADDJS93}
I.~Alth\"ofer, G.~Das, D.~Dobkin, D.~Joseph, and J.~Soares.
\newblock On sparse spanners of weighted graphs.
\newblock {\em Discrete Comput. Geom.}, 9:81--100, 1993.

\bibitem[ADF{\etalchar{+}}17]{ADFSW17}
Stephen Alstrup, Soren Dahlgaard, Arnold Filtser, Morten Stockel, and Christian
  Wulff-Nilsen.
\newblock Personal communication, 2017.

\bibitem[AHL02]{AHL02}
Noga Alon, Shlomo Hoory, and Nathan Linial.
\newblock The moore bound for irregular graphs.
\newblock {\em Graphs and Combinatorics}, 18(1):53--57, 2002.

\bibitem[AN12]{AN12}
Ittai Abraham and Ofer Neiman.
\newblock Using petal-decompositions to build a low stretch spanning tree.
\newblock In {\em STOC}, pages 395--406, 2012.

\bibitem[AP92]{AP92}
B.~Awerbuch and D.~Peleg.
\newblock Routing with polynomial communication-space tradeoff.
\newblock {\em SIAM J. Discrete Mathematics}, 5:151--162, 1992.

\bibitem[Awe85]{Awe85}
B.~Awerbuch.
\newblock Complexity of network synchronization.
\newblock {\em J. ACM}, 4:804--823, 1985.

\bibitem[Bar96]{Bar96}
Yair Bartal.
\newblock Probabilistic approximations of metric spaces and its algorithmic
  applications.
\newblock In {\em FOCS}, pages 184--193, 1996.

\bibitem[Bar98]{B98}
Yair Bartal.
\newblock On approximating arbitrary metrices by tree metrics.
\newblock In {\em Proceedings of the Thirtieth Annual ACM Symposium on Theory
  of Computing}, STOC '98, pages 161--168, New York, NY, USA, 1998. ACM.

\bibitem[Bar04]{B04}
Yair Bartal.
\newblock Graph decomposition lemmas and their role in metric embedding
  methods.
\newblock In {\em Algorithms - {ESA} 2004, 12th Annual European Symposium,
  Bergen, Norway, September 14-17, 2004, Proceedings}, pages 89--97, 2004.

\bibitem[Bas08]{B08}
S.~Baswana.
\newblock Streaming algorithm for graph spanners - single pass and constant
  processing time per edge.
\newblock {\em Inf. Process. Lett.}, 106(3):110--114, 2008.

\bibitem[BCE05]{BCE05}
B{\'{e}}la Bollob{\'{a}}s, Don Coppersmith, and Michael Elkin.
\newblock Sparse distance preservers and additive spanners.
\newblock {\em {SIAM} J. Discrete Math.}, 19(4):1029--1055, 2005.

\bibitem[BGK{\etalchar{+}}14]{BGKMPT11}
Guy~E. Blelloch, Anupam Gupta, Ioannis Koutis, Gary~L. Miller, Richard Peng,
  and Kanat Tangwongsan.
\newblock Nearly-linear work parallel {SDD} solvers, low-diameter
  decomposition, and low-stretch subgraphs.
\newblock {\em Theor. Comp. Sys.}, 55(3):521--554, October 2014.

\bibitem[BKMP10]{BKMP10}
Surender Baswana, Telikepalli Kavitha, Kurt Mehlhorn, and Seth Pettie.
\newblock Additive spanners and (alpha, beta)-spanners.
\newblock {\em {ACM} Transactions on Algorithms}, 7(1):5, 2010.

\bibitem[BR10]{BR10}
Ajesh Babu and Jaikumar Radhakrishnan.
\newblock An entropy based proof of the moore bound for irregular graphs.
\newblock {\em CoRR}, abs/1011.1058, 2010.

\bibitem[BS03]{BS03}
S.~Baswana and S.~Sen.
\newblock A simple linear time algorithm for computing a $(2k-1)$-spanner of
  ${O}(n^{1+1/k})$ size in weighted graphs.
\newblock In {\em Proceedings of the 30th International Colloquium on Automata,
  Languages and Programming}, volume 2719 of {\em LNCS}, pages 384--396.
  Springer, 2003.

\bibitem[BS07]{BS07}
Surender Baswana and Sandeep Sen.
\newblock A simple and linear time randomized algorithm for computing sparse
  spanners in weighted graphs.
\newblock {\em Random Struct. Algorithms}, 30(4):532--563, 2007.

\bibitem[BW15]{BW15}
Gregory Bodwin and Virginia~Vassilevska Williams.
\newblock Very sparse additive spanners and emulators.
\newblock In {\em Proceedings of the 2015 Conference on Innovations in
  Theoretical Computer Science, {ITCS} 2015, Rehovot, Israel, January 11-13,
  2015}, pages 377--382, 2015.

\bibitem[Che13]{C13}
Shiri Chechik.
\newblock New additive spanners.
\newblock In {\em Proceedings of the Twenty-Fourth Annual ACM-SIAM Symposium on
  Discrete Algorithms}, SODA '13, pages 498--512, Philadelphia, PA, USA, 2013.
  Society for Industrial and Applied Mathematics.

\bibitem[Coh99]{Coh93}
E.~Cohen.
\newblock Fast algorithms for $t$-spanners and stretch-$t$ paths.
\newblock {\em SIAM J. Comput.}, 28:210--236, 1999.

\bibitem[Coh00]{C00}
Edith Cohen.
\newblock Polylog-time and near-linear work approximation scheme for undirected
  shortest paths.
\newblock {\em J. {ACM}}, 47(1):132--166, 2000.

\bibitem[DGPV08]{DGPV08}
Bilel Derbel, Cyril Gavoille, David Peleg, and Laurent Viennot.
\newblock On the locality of distributed sparse spanner construction.
\newblock In {\em Proceedings of the Twenty-Seventh Annual {ACM} Symposium on
  Principles of Distributed Computing, {PODC} 2008, Toronto, Canada, August
  18-21, 2008}, pages 273--282, 2008.

\bibitem[DHZ00]{DHZ00}
D.~Dor, S.~Halperin, and U.~Zwick.
\newblock All-pairs almost shortest paths.
\newblock {\em SIAM J. Comput.}, 29:1740--1759, 2000.

\bibitem[DMP{\etalchar{+}}03]{DMPRS03}
D.~Dubhashi, A.~Mei, A.~Panconesi, J.~Radhakrishnan, and A.~Srinivisan.
\newblock Fast distributed algorithm for (weakly) connected dominating sets and
  linear-size skeletons, 2003.

\bibitem[DMZ06]{DMZ06}
Bilel Derbel, Mohamed Mosbah, and Akka Zemmari.
\newblock Fast distributed graph partition and application.
\newblock In {\em 20th International Parallel and Distributed Processing
  Symposium {(IPDPS} 2006), Proceedings, 25-29 April 2006, Rhodes Island,
  Greece}, 2006.

\bibitem[EEST05]{EEST05}
Michael Elkin, Yuval Emek, Daniel~A. Spielman, and Shang-Hua Teng.
\newblock Lower-stretch spanning trees.
\newblock In {\em STOC}, pages 494--503, 2005.

\bibitem[Elk04]{E04}
M.~Elkin.
\newblock An unconditional lower bound on the time-approximation tradeoff of
  the minimum spanning tree problem.
\newblock In {\em Proc. of the 36th ACM Symp. on Theory of Comput. (STOC
  2004)}, pages 331--340, 2004.

\bibitem[Elk05]{E05}
Michael Elkin.
\newblock Computing almost shortest paths.
\newblock {\em {ACM} Trans. Algorithms}, 1(2):283--323, 2005.

\bibitem[Elk07a]{E06}
Michael Elkin.
\newblock A near-optimal distributed fully dynamic algorithm for maintaining
  sparse spanners.
\newblock In {\em Proceedings of the Twenty-Sixth Annual {ACM} Symposium on
  Principles of Distributed Computing, {PODC} 2007, Portland, Oregon, USA,
  August 12-15, 2007}, pages 185--194, 2007.

\bibitem[Elk07b]{E07}
Michael Elkin.
\newblock Streaming and fully dynamic centralized algorithms for constructing
  and maintaining sparse spanners.
\newblock In {\em Automata, Languages and Programming, 34th International
  Colloquium, {ICALP} 2007, Wroclaw, Poland, July 9-13, 2007, Proceedings},
  pages 716--727, 2007.

\bibitem[EN16a]{EN16b}
Michael Elkin and Ofer Neiman.
\newblock Distributed strong diameter network decomposition.
\newblock In {\em Proceedings of the 2016 {ACM} Symposium on Principles of
  Distributed Computing, {PODC} 2016, Chicago, IL, USA, July 25-28, 2016},
  pages 211--216, 2016.

\bibitem[EN16b]{EN16c}
Michael Elkin and Ofer Neiman.
\newblock Hopsets with constant hopbound, and applications to approximate
  shortest paths.
\newblock In {\em {IEEE} 57th Annual Symposium on Foundations of Computer
  Science, {FOCS} 2016, 9-11 October 2016, Hyatt Regency, New Brunswick, New
  Jersey, {USA}}, pages 128--137, 2016.

\bibitem[EN17]{EN17}
Michael Elkin and Ofer Neiman.
\newblock Efficient algorithms for constructing very sparse spanners and
  emulators.
\newblock In {\em Proceedings of the Twenty-Eighth Annual ACM-SIAM Symposium on
  Discrete Algorithms}, pages 652--669, 2017.

\bibitem[EP04]{EP04}
Michael Elkin and David Peleg.
\newblock (1+epsilon, beta)-spanner constructions for general graphs.
\newblock {\em {SIAM} J. Comput.}, 33(3):608--631, 2004.

\bibitem[EP15]{EP15}
Michael Elkin and Seth Pettie.
\newblock A linear-size logarithmic stretch path-reporting distance oracle for
  general graphs.
\newblock In {\em Proceedings of the Twenty-Sixth Annual {ACM-SIAM} Symposium
  on Discrete Algorithms, {SODA} 2015, San Diego, CA, USA, January 4-6, 2015},
  pages 805--821, 2015.

\bibitem[ES16]{ES16}
Michael Elkin and Shay Solomon.
\newblock Fast constructions of lightweight spanners for general graphs.
\newblock {\em {ACM} Trans. Algorithms}, 12(3):29:1--29:21, 2016.

\bibitem[EZ06]{EZ06}
Michael Elkin and Jian Zhang.
\newblock Efficient algorithms for constructing (1+epsilon, beta)-spanners in
  the distributed and streaming models.
\newblock {\em Distributed Computing}, 18(5):375--385, 2006.

\bibitem[FKM{\etalchar{+}}05]{FK+05}
J.~Feigenbaum, S.~Kannan, A.~McGregor, S.~Suri, and J.~Zhang.
\newblock Graph distances in the streaming model: The value of space.
\newblock In {\em Proc. of the ACM-SIAM Symp. on Discrete Algorithms}, pages
  745--754, 2005.

\bibitem[FS16]{FS16}
Arnold Filtser and Shay Solomon.
\newblock The greedy spanner is existentially optimal.
\newblock In {\em Proceedings of the 2016 {ACM} Symposium on Principles of
  Distributed Computing, {PODC} 2016, Chicago, IL, USA, July 25-28, 2016},
  pages 9--17, 2016.

\bibitem[HZ96]{HZ96}
S.~Halperin and U.~Zwick.
\newblock Linear time deterministic algorithm for computing spanners for
  unweighted graphs, 1996.
\newblock manuscript.

\bibitem[KMP11]{KMP11}
Ioannis Koutis, Gary~L. Miller, and Richard Peng.
\newblock A nearly-m log n time solver for sdd linear systems.
\newblock In {\em Proceedings of the 2011 IEEE 52Nd Annual Symposium on
  Foundations of Computer Science}, FOCS '11, pages 590--598, Washington, DC,
  USA, 2011. IEEE Computer Society.

\bibitem[LS93]{LS93}
N.~Linial and M.~Saks.
\newblock Decomposing graphs into regions of small diameter.
\newblock {\em Combinatorica}, 13:441--454, 1993.

\bibitem[MPVX15]{MPVX15}
Gary~L. Miller, Richard Peng, Adrian Vladu, and Shen~Chen Xu.
\newblock Improved parallel algorithms for spanners and hopsets.
\newblock In {\em Proceedings of the 27th ACM Symposium on Parallelism in
  Algorithms and Architectures}, SPAA '15, pages 192--201, New York, NY, USA,
  2015. ACM.

\bibitem[MPX13]{MPX13}
Gary~L. Miller, Richard Peng, and Shen~Chen Xu.
\newblock Parallel graph decompositions using random shifts.
\newblock In {\em 25th {ACM} Symposium on Parallelism in Algorithms and
  Architectures, {SPAA} '13, Montreal, QC, Canada - July 23 - 25, 2013}, pages
  196--203, 2013.

\bibitem[OMSW10]{OMSW10}
James~B. Orlin, Kamesh Madduri, K.~Subramani, and M.~Williamson.
\newblock A faster algorithm for the single source shortest path problem with
  few distinct positive lengths.
\newblock {\em J. of Discrete Algorithms}, 8:189--198, June 2010.

\bibitem[Pel99]{P99}
David Peleg.
\newblock Proximity-preserving labeling schemes and their applications.
\newblock In {\em Graph-Theoretic Concepts in Computer Science, 25th
  International Workshop, {WG} '99, Ascona, Switzerland, June 17-19, 1999,
  Proceedings}, pages 30--41, 1999.

\bibitem[Pel00]{P00}
David Peleg.
\newblock {\em Distributed Computing: A Locality-sensitive Approach}.
\newblock Society for Industrial and Applied Mathematics, Philadelphia, PA,
  USA, 2000.

\bibitem[Pet09]{Pet09}
Seth Pettie.
\newblock Low distortion spanners.
\newblock {\em ACM Transactions on Algorithms}, 6(1), 2009.

\bibitem[Pet10]{P10}
Seth Pettie.
\newblock Distributed algorithms for ultrasparse spanners and linear size
  skeletons.
\newblock {\em Distributed Computing}, 22(3):147--166, 2010.

\bibitem[Pet16]{Pettie-privcomm}
Seth Pettie.
\newblock Personal communication, 2016.

\bibitem[PS89]{PS89}
D.~Peleg and A.~Sch\"affer.
\newblock Graph spanners.
\newblock {\em J. Graph Theory}, 13:99--116, 1989.

\bibitem[PU89a]{PU89a}
D.~Peleg and J.~D. Ullman.
\newblock An optimal synchronizer for the hypercube.
\newblock {\em SIAM J. on Comput.}, 18:740--747, 1989.

\bibitem[PU89b]{PU89b}
D.~Peleg and E.~Upfal.
\newblock A tradeoff between size and efficiency for routing tables.
\newblock {\em J. of the ACM}, 36:510--530, 1989.

\bibitem[Rei93]{R93}
John~H. Reif.
\newblock {\em Synthesis of Parallel Algorithms}.
\newblock Morgan Kaufmann Publishers Inc., San Francisco, CA, USA, 1st edition,
  1993.

\bibitem[RZ04]{RZ04}
Liam Roditty and Uri Zwick.
\newblock On dynamic shortest paths problems.
\newblock In {\em Algorithms - {ESA} 2004, 12th Annual European Symposium,
  Bergen, Norway, September 14-17, 2004, Proceedings}, pages 580--591, 2004.

\bibitem[TZ05]{TZ01}
Mikkel Thorup and Uri Zwick.
\newblock Approximate distance oracles.
\newblock {\em J. {ACM}}, 52(1):1--24, 2005.

\bibitem[TZ06]{TZ06}
M.~Thorup and U.~Zwick.
\newblock Spanners and emulators with sublinear distance errors.
\newblock In {\em Proc. of Symp. on Discr. Algorithms}, pages 802--809, 2006.

\bibitem[Woo06]{W06}
David~P. Woodruff.
\newblock Lower bounds for additive spanners, emulators, and more.
\newblock In {\em Proceedings of the 47th Annual IEEE Symposium on Foundations
  of Computer Science}, FOCS '06, pages 389--398, Washington, DC, USA, 2006.
  IEEE Computer Society.

\end{thebibliography}

\end{document}